\documentclass[10pt,twocolumn,showpacs,superscriptaddress,notitlepage,jcp,aip,amsmath,amssymb]{revtex4-1}
\usepackage{epsfig}

\newcommand{\R}{\ensuremath{\mathbb {R}}}

\newcommand{\EE}{\ensuremath{\mathbb {E}}}

\newcommand{\s}{\sigma}
\newcommand{\ideal}{{\mathrm{ideal}}}
\newcommand{\Poisson}{{\mathrm{Poisson}}}

\newcommand{\dvol}{d \mathrm{vol}}
\newcommand{\Vol}{\mathrm{Vol}}
\newcommand{\vect}[1]{{\begin{bmatrix}#1\end{bmatrix}}}
\newcommand{\scp}[1]{\langle #1\rangle} 
\renewcommand{\vec}[1]{\mathbf{#1}}
\newcommand{\cc}[1]{Ref.~\onlinecite[#1]{SchneiderWeilBook}}
\newcommand{\Conv}[1]{\mathrm{Conv}(#1)}
\newcommand{\arccosh}{\mathrm{arccosh}}

\newenvironment{proof}{\topsep=\smallskipamount \partopsep=0pt  %
\begin{trivlist} \itemindent=\parindent                        %
 \item[\hskip \labelsep\emph{Proof:}]}{\qed\end{trivlist}}   
\let\qed=\relax                                                 %
\def\qed                                                        %
 {{\unskip\nobreak\hfil\penalty50                               %
   \quad\hbox{}\nobreak\hfil $\Box$                             %
   \parfillskip=0pt \finalhyphendemerits=0 \par}}  

\newtheorem{theorem}{Theorem}                          %
\newtheorem{proposition}[theorem]{Proposition}                  %
\newtheorem{lemma}[theorem]{Lemma}

\begin{document}
\title{Geometrical Frustration and Static Correlations in Hard-Sphere Glass Formers}
\author{Benoit~Charbonneau}
\affiliation{Mathematics Department, St.Jerome's University in the University of Waterloo, Waterloo, Ontario, N2L 3G3, Canada}
\author{Patrick~Charbonneau}
\affiliation{Department of Chemistry, Duke University, Durham,
North Carolina 27708, USA}
\affiliation{Department of Physics, Duke University, Durham,
North Carolina 27708, USA}
\author{Gilles~Tarjus}
\affiliation{LPTMC, CNRS-UMR 7600, Universit\'e Pierre et Marie Curie, bo\^ite 121,
4 Place Jussieu, 75005 Paris, France}

\date{October 12, 2012.  Modified: December 21, 2012.}

\begin{abstract}
We analytically and numerically characterize the structure of hard-sphere fluids in order to review various geometrical frustration scenarios of the glass transition. We find generalized polytetrahedral order to be correlated with increasing fluid packing fraction, but to become increasingly irrelevant with increasing dimension. We also find the growth in structural correlations to be modest in the dynamical regime accessible to computer simulations.
\end{abstract}


\maketitle
Two key problems in glass physics can be formulated as follows: \emph{(i)} What is the mechanism that thwarts crystallization of a liquid and 
allows for the development of glassy behavior? \emph{(ii)} Can one assign the fast increase of the relaxation time and the viscosity as one approaches the 
glass transition to a collective phenomenon and the growth of static correlations? Possible answers to both questions have been proposed, 
based on the concept of  ``frustration''. Frustration in this context refers to an incompatibility between some local order that is preferred in the fluid and 
a global periodic tiling of space~\cite{sadoc:1999,nelson:2002,tarjus:2005}.

A detailed description of frustration requires knowledge of the relevant type of local order, a 
difficult (and unsolved) task in the case of a generic molecular glass-forming liquid.  Hard-sphere fluids however are well suited for the task, 
because they are minimal model glass formers and their geometry is unambiguous. Frustration is then purely geometrical. For monodisperse and moderately 
polydisperse hard-sphere fluids, the locally preferred order is based on tetrahedral arrangements of 4 particles, which when combined in larger units leads to icosahedral, 
or more rigorously polytetrahedral, order whose extension to the whole space is indeed frustrated. As often in physics, insight is gained by enlarging the 
parameter space. Geometrical frustration can be varied and fruitfully studied by varying the properties of the embedding physical space, either by 
introducing curvature~\cite{sadoc:1999,nelson:2002,sausset:2008,sausset:2010} or by changing the number of dimensions. We follow here the latter route. 

Three-dimensional tetrahedral order generalizes to ``simplicial order'' in all dimensions, the simplex being the ideal packing of $d+1$ particles in $d$ 
dimensions: a triangle in 2 dimensions, a tetrahedron in 3 dimensions, a ``hyper-tetrahedron'' in higher dimensions. In $d=2$, the simplicial order 
is not frustrated, but it is frustrated in $d=3$, and even more so in higher dimensions. Dimensional studies have already unambiguously 
answered the above question \emph{(i)} in the case of hard-sphere systems~\cite{vanmeel:2009,vanmeel:2009b}, by showing that the local orders of the fluid and 
the crystal become increasingly different 
as dimension increases and that, as a result of this frustration, crystallization is strongly suppressed in dimensions 
$d\geq 4$. Whereas crystallization kinetically interferes in three-dimensional monodisperse hard-sphere fluids~\cite{auer:2001}, 
so that one needs  to consider binary or polydisperse mixtures in order to study the metastable fluid branch~\cite{foffi:2003,berthier:2009,flenner:2011}, 
going to higher dimensions sidesteps this particular issue by suppressing crystallization even in monodisperse systems.

The question that we address in this paper is then the second one, namely the connection between the dynamical slowdown and the growth of static structural 
correlations, which result in the super-Arrhenius, or fragile, dynamical behavior of glass-forming fluids. If indeed the phenomenon is collective, it should be accompanied by the development of nontrivial correlations to which 
one should be able to associate one or several typical length scales. It has been realized in the past two 
decades that slow relaxation in glass formers has an increasingly heterogeneous character, which can be captured by 
an appropriately defined ``dynamical'' length obtained through multi-point space-time correlation functions~\cite{berthier:2011b}. 
The somewhat different issue that we consider here is whether the slowing down itself, as described by the increase of the structural relaxation time or the viscosity, 
is accompanied by the growth of a static length. 
 
There has been a number of proposals for relating dynamics and structure that include some measure of the spatial extension of a characteristic locally 
preferred arrangement~\cite{steinhardt:1983,dzugutov:2002,miracle:2004,tarjus:2005,sausset:2008,sausset:2010,coslovich:2007,coslovich:2011,
kawasaki:2007,shintani:2008,watanabe:2008,leocmach:2012,malins:2012}. We investigate in the present work various measures of frustrated 
simplicial order in hard-sphere fluids in dimensions ranging from 3 to 8. Our results, which combine analytical and numerical approaches, show that the 
growth of associated structural correlations under compression is very limited in the density range that is accessible to computer simulation. This observation
can be rationalized by invoking the strong degree of frustration that is present, even in $d=3$, and that limits the putative collective behavior due to the extension 
of the frustrated local order. It could of course be that the structural correlations related to simplicial order do not capture the relevant static correlations that 
affect the dynamical slowdown. We therefore also investigate ``order-agnostic'' lengths, which characterize static point-to-set correlations; they are 
expected to provide upper bounds to the more specific length scales associated with local order while being more directly related to the cooperative nature of the 
dynamics~\cite{montanari:2006}. We also find that the evolution with density of such point-to-set lengths is quite limited in a domain where the 
relaxation time nonetheless increases by several orders of magnitude and the length associated with the dynamical heterogeneity grows by a factor of 4 or 5.

The paper is organized as follows. In sections~\ref{sect:geobackgnd} and~\ref{sect:geofrus}, we review notions of higher-dimensional simplicial geometry as well as various geometrical proposals for linking the local geometrical order and the dynamical slowdown of fluids, respectively. In section~\ref{sect:liquidorder}, we obtain analytical bounds for the fluid order and numerically characterize the simplicial structure of hard-sphere fluids in $d$=3 to 8. In section~\ref{sect:liquidagnostic}, we describe features of an ``order-agnostic'' structural length, calculate its growth for two hard-spheres fluids in $d$=3, and compare the results with a bound extracted from the fluid dynamics. A brief conclusion follows.

\section{Geometrical Tilings of simplices}
\label{sect:geobackgnd}

In order to describe ideal tilings of hard spheres we need to review some notions of Euclidean space $\R^d$, hyperbolic space $H^d$, and spherical space $S^d$ geometry. Note that the sphere space of (surface) dimension $d$, $S^d$, is defined by the set of all vectors of equal length in $\R^{d+1}$.   Hence $S^2$ is the standard sphere in $\R^3$. Consider a (non-overlapping) configuration of $N$ hard spheres in $d$-dimensional space of volume $V$~\cite{footnote:1}, at positions $\vec{r}_1,\ldots,\vec{r}_N$.  The neighbors of a particle are defined by the Delaunay tessellation associated with those positions, and the distances are measured in units of hard-sphere diameter $\sigma$. In $\R^d$, a configuration therefore has a number density $\rho=N/V$ that is related to its packing fraction  $\varphi$ via
\begin{equation}
\varphi=V_{d}(\sigma/2)\rho,
\end{equation}
where $V_{d}(R)=R^d\pi^{d/2}/\Gamma(d/2+1)$ is the volume of a $d$-dimensional ball of radius $R$.

A simplex contains $d+1$ vertices, and each of its faces is itself a lower-dimensional simplex (triangle in $d$=2, tetrahedron in $d$=3, \ldots). It is thus the simplest polytope in any dimension. It also plays a central role in our analysis because, with probability one, simplices are the basic tiles of a Delaunay tessellation.  Consider the $i$-simplices  $\s_i(1),\ldots,\s_i(N_i)$ of a Delaunay tessellation.  Of course, $N_0=N$ is the number of particles, $N_1$  is the number of bonds, and in general 
\begin{equation}
\sum_{i=0}^d (-1)^iN_i=\chi,
\end{equation}
where the Euler characteristic, $\chi$, depends on the topology of the space considered. 
Let 
\[q_{j,d}(i)\equiv\# \s_d \text{ wrapped around }\s_j(i),\]
and
\[\bar{q}_{j,d}\equiv \frac1{N_j}\sum_{i=1}^{N_j}q_{j,d}(i)\]
be its average. We thus have that $q_{0,2}$ is the number of triangles sharing a common point, and that $\bar{q}_{0,2}$ is the average number of triangles per vertex in a configuration. In this work, we pay particular attention to \emph{bond spindles}~\cite{footnote:2}, $\bar{q}_{d-2,d}$.  We also consider $Z_d(j)$, the coordination number of the $j^{th}$ vertex, i.e., $Z_d(j)$ is the number of $\s_1$ wrapped around $\s_0(j)$, and therefore
\begin{equation*}
	\bar Z_d=\frac{1}{N}\sum_{i=1}^N Z_d(i).
\end{equation*}

In order to describe the regular tessellations of higher-dimensional space, we introduce Schl\"afli's notation. In this recursive description of generalized regular polyhedra, the polytope $\{a_1,\ldots,a_d\}$ is composed of polytopes $\{a_1,\ldots,a_{d-1}\}$ with vertex figure $\{a_2,\ldots,a_{d}\}$ (see Ref.~\onlinecite[Chapter 7]{coxeter:1973}). The \emph{vertex figure} of a $d$-dimensional polytope is a $(d-1)$-dimensional polytope whose vertices are obtained by taking the middle of each edge emanating from a vertex.  
The basic example of this inductive scheme is $\{p\}$, the regular polygon of $p$ sides in $\R^2$. We then have the regular polytopes in $\R^3$, which are represented by $\{p,q\}$ for a polytope having faces made of $\{p\}$'s and vertex figures made of $\{q\}$'s.   For instance, the cube $\{4,3\}$ is composed of squares (the regular polygon $\{4\}$) and the vertex figure is $\{3\}$, a triangle.   The cubic tiling $\{4,4\}$ in $\R^2$ is composed of 4 squares surrounding each vertex. Note also that simplices are denoted $\{3^{d-1}\}$.  (In Schl\"afli symbols, exponentiation is understood to mean repetition, not multiplication.) 
Of particular relevance to this paper is the relation between Schl\"afli's notation and the spindle coordination.  For $\{a_1,\ldots,a_d\}$ that is a regular tiling in $\R^d$ or  $S^d$, one can indeed show that $q_{d-2,d}=a_d$ (Appendix~\ref{sect:appschaefli}). 

Any polytope $\{a_1,\ldots,a_d\}$ in $\R^{d+1}$ can be seen as inscribed in a sphere $S^d$ whose center is the center of mass of the polytope.  Inflating the polytope so as to make it round produces a regular tiling of $S^d$.  For instance, the hypercube $\{4,3^{d-2}\}$ in $\R^{d}$ is also a regular tiling of $S^{d-1}$.  The usual case $\{4,3\}$ therefore gives $q_{0,2}=3$ on $S^2$.   The same argument illustrates that for the tiling $\{4,3^{d-1}\}$ of $S^{d}$ one has $q_{d-2,d}=3$. 

We stress that in the case of hard-sphere fluids, we are interested in tilings by simplices. A detailed review of the possible regular simplicial tilings in various spaces is available in the Supplementary Material of Ref.~\onlinecite{charbonneau:2012}. Here, we only report the curvature of the space that would embed these tilings (Table~\ref{table:summary}), and highlight a few key observations.
\begin{itemize}
\item A regular tiling of tetrahedra is found on the sphere that inscribes the remarkably large $d=4$ regular polytope $\{3,3,5\}$.  The inscribing sphere's curvature is  much smaller than that of the generalized octahedron $\{3,3,4\}$, the only other polytope exclusively made of simplices.
\item By contrast, on the $d$-dimensional sphere with $d>3$,  $\{3^{d-1},4\}$ is the only way to obtain a regular tiling of simplices, which can only be achieved on a sphere of much higher curvature. 
\item The regular simplicial tiling $\{3,3,3,6\}$ is found on a hyperbolic space of curvature larger (in absolute value) than the curvature of the sphere that inscribes the tiling $\{3,3,3,4\}$, while the hyperbolic tiling $\{3,3,6\}$'s infinite edge length disqualifies it as a possible reference for Euclidean packings of spheres. 
\end{itemize}

It is also interesting to note that the densest known crystal lattice in $d$=8, the root lattice $E_8$, is singularly dense, which arises from it having a remarkably high fraction of regular simplices in its midst~\cite{conway:1988}. Each of the sites in the $E_8$ lattice is indeed surrounded by 2160 generalized octahedra and 17280 simplexes.

\begin{table}
\begin{center}
\begin{tabular}{|c|c|c|}
	\hline
	Simplicial tilings & in & curvature\\
	\hline
	$\{3,6\}$ & $\R^2$&0\\
	$\{3,n\}$ with $n\geq 7$ & $H^2$ & $-4\arccosh^2(2^{-1}\csc\frac\pi n)$\\
	600-cell $\{3,3,5\}$ & $S^3$ & ${\pi^2}/{25}\simeq 0.39$\\
	octahedron $\{3^{d-1},4\}$ & $S^d$ & ${\pi^2}/{4}\simeq 2.47$\\
	$\{3,3,3,6\}$ & $H^4$ & $-4\arccosh^2(\frac{1+\sqrt5}2)\simeq -4.505$\\
	\hline
\end{tabular}
\end{center}
\caption{Summary of the regular simplicial tilings in $d\geq2$. The curvature reported in units of $\sigma$.}
\label{table:summary}
\end{table}

\section{Formulations of geometrical frustration}
\label{sect:geofrus}
Various geometrical proposals have been made to explain the growth of a static correlation length in glass-forming fluids. These approaches depend on the geometrical frustration that arises when optimally local and global orders are incompatible. All of these approaches thus also rely, in one form or another,  on an analogy with hard-disk packings in $\R^2$. Thue long ago showed that the densest possible $\R^2$ packing of disks is the triangular lattice~\cite{thue:1892}. Interestingly, both the densest closed-shell cluster -- seven disks forming a hexagon -- and the perfect simplex -- three disks forming an equilateral triangle -- tile space into this lattice. This system is thus understood having no geometrical frustration. In fact even the thermodynamic transitions between the fluid and crystal are either continuous or very weakly first order~\cite{halperin:1978,bernard:2011}. 

The behavior of hard disks in $\R^2$ contrasts with that of disks on negatively curved space~\cite{sausset:2008d,sausset:2010}, of mixtures of different size disks and spheres~\cite{kawasaki:2007,flenner:2011,foffi:2003,berthier:2009,kranendonk:1991,auer:2001b}, or of spheres in higher-dimensional spaces~\cite{vanmeel:2009,vanmeel:2009b}. In that sense, geometrical frustration between the fluid and the crystalline order is ubiquitous. By severely inhibiting crystallization, geometrical frustration extends the fluid branch to deeply supersaturated fluids, which gives access to the dynamically sluggish regime and glass formation~\cite{sausset:2008d,charbonneau:2010,charbonneau:2012b}. 

We also wish to determine if the dynamical slowdown of these supersaturated fluids has a straightforward geometrical origin. Can simple geometry provide both a measure of a growing static length and a mechanism for the slowdown? Higher dimensions provide an additional constraint. Because the glassy phenomenology of hard-sphere fluids is remarkably robust to dimensional changes~\cite{charbonneau:2010,charbonneau:2012b}, 
we expect a geometrical proposal to be generalizable to arbitrary dimensions. We use this dimensional requirement to critically review various proposals for the geometrical origin to the glass transition. 

\subsection{Disclination Lines}
Regular simplices are the densest structure for packing hard spheres in any dimension. In dense fluids, where free volume is an important contribution to the fluid free energy, one would therefore expect simplices to have a growing structural representation~\cite{anikeenko:2007,anikeenko:2008}. For Euclidean space in dimension $d\geq 3$, however, perfect simplices cannot tile space without defects. In $d=3$, a mix of tetrahedra and octahedra is needed to assemble the face-centered cubic lattice; in $d=4$, only octahedra are needed to generate the densest known lattice~\cite{conway:1988}, $D_4$, and kissing cluster.

Shortly after noticing that a perfect tiling of simplices with $q_{1,3}=5$ for all bonds is possible in $d=3$ on the relatively gently curved $S^3$ space that embeds the $\{3,3,5\}$ polytope (Table~\ref{table:summary}), it was suggested that $d=3$ fluid configurations can be usefully described as frustrated with respect to that order. The fluid should then be understood as weakly defective with respect to that perfect tiling~\cite{nelson:2002,sadoc:1999}. The structural defects that result from ``flattening''   $S^3$ tilings to $\R^3$ are known as disclinations. They are best understood by analogy to $d=2$ tilings. Each particle in a perfect $d=2$ triangular lattice of disks on $\R^2$ is part of six triangles, i.e., $q_{0,2}=6$. Curving that plane results in irreducible disclinations that sit on disk centers, and whose coordination differs from six. Irreducible disclinations with $q_{0,2}>6$ are obtained in spaces of negative curvature~\cite{sausset:2008}, e.g., $H^2$,  and with $q_{0,2}<6$ in spaces with positive curvature, e.g., $S^2$ in Ref.~\onlinecite{nelson:2002}. In $d=3$, these disclinations analogously correspond to bond spindles. Uncurving the $S^3$ space that embeds the $\{3,3,5\}$ thus results in irreducible bond spindles with $q_{1,3}>5$.

Periodic arrangements of $q_{1,3}$ bond spindles form the complex Frank--Kasper crystal structures~\cite{frank:1959}, which are the densest packings of certain binary hard sphere mixtures~\cite{kummerfeld:2008,filion:2009,hopkins:2012}. Some of these structures even come close to approximating the ideal packing bound for $\bar{Z}_d$ (see Section~\ref{sect:idealpacking}). In \emph{weakly frustrated} fluid configurations in $d$=3, rules for disclination propagation suggest that topology constrains the crossing of disclination lines. This effect could contribute to the dynamical slowdown, because breaking a constraint is an activated event~\cite{nelson:2002}. In this scenario, increasing the fluid density brings the system closer to a perfect tiling of simplices, giving a growing importance to these topological constraints, which further slows the fluid dynamics. In other words, this framework suggests that there should exist a causality between the dynamical slowdown and a growing static, structural correlation length associated with the bond spindle order, which would be an explanation for dynamical fragility in glass-forming fluids. Note that small frustration is also crucial for the avoided critical scenario of the frustration-limited domain theory~\cite{tarjus:2005}.

The small frustration assumption in this argument is crucial. Higher-order spindle defects are not necessarily constrained by the same topological rules, and these constraints may go away altogether. Yet regular spaces with $d>3$ provide no higher-dimensional equivalent to the $\{3,3,5\}$. As detailed in Section~\ref{sect:geobackgnd}, regular packings of simplices in $d>3$, are only found on $S^d$ with a curvature that is more than six times that for embedding $\{3,3,5\}$ in $d$=3 (Table~\ref{table:summary}). Because even in $d$=3 the concentration of defects is so high that each particle has multiple defective bond spindles~\cite{charbonneau:2012}, simplex ordering in $d>3$ fluids should therefore be quite limited. Yet it could be that it is not the concentration of defects that matters, but their structural correlation, so it remains possible for $d=3$ to be a limit case, where this geometrical frustration variant remains in effect. 

\subsection{Percolating spindle order}
An alternative explanation for the dynamical slowdown in $d=3$ invokes the percolation of icosahedral clusters~\cite{tomida:1995}. These icosahedra are centered on particles for which all bond spindles are ideal, i.e., $q_{1,3}=5$, which result in assemblies of nearly perfect simplices. Once these dense and presumably stable structures percolate, they have been proposed to affect the dynamical evolution of the rest of the system. Because the fraction of $q_{1,3}=5$ grows with packing fraction, the onset of icosahedral percolation should then correspond to the onset of sluggish dynamics. And because the percolating structure is be expected to become more robust with density, this process could also be an explanation for the dynamical fragility of certain glass-formers.

Generalizing this approach to higher dimensions encounters two main difficulties. First, there are no equivalents to the closed-shell assembly of simplices into a generalized icosahedra in higher $d$. One could argue that the fraction of $q_{d-2,d}=5$ or 4 should be considered instead, but then one runs into a second difficulty. Because the onset of percolation rapidly decreases with dimension~\cite{Grassberger:2003}, systems of all densities, even the ideal gas, would be have a percolating order. It thus seems that this scenario does not apply to higher $d$, although $d$=3 could still be a limit case.

\subsection{Medium-Range Crystalline Order}
Another variant of the schemes above has been advanced by Tanaka and coworkers~\cite{kawasaki:2007,shintani:2008,watanabe:2008,leocmach:2012}, who propose that a growing presence of medium-range crystalline order (MRCO) underlies the structural origin of the dynamical slowdown. The presence of MRCO would straightforwardly explain the dynamical slowdown because crystal-like arrangements are intrinsically dynamically more stable. The more of these structural features there are, the slower the overall dynamics should be. Some form of frustration, however, must be invoked to explain why MRCO does not extend to the whole space, although the fluid is in a domain where the crystal is the most stable phase.

Observations in $d=2$ (Ref.~\onlinecite{watanabe:2008}), and in moderately polydisperse systems in $d=3$ (Ref.~\onlinecite{leocmach:2012}), have found that the growth of MRCO can be correlated with the dynamical slowdown. From a dimensional point of view, however, this proposal is indefensible. Studies in $d>3$ have found that the fluid order grows increasingly different from the crystal order, and that  no trace of crystal order centercan be found even in monodisperse systems where the crystal is unique and singularly dense~\cite{charbonneau:2010}. We confirm this result below using a different approach, but once again, this description cannot apply to higher $d$, although $d$=3 may be a limit case.

\subsection{Generalized Preferred Cluster}
Coslovich et al. have proposed to generalize the above description by systematically investigating the local environment of the particles in a given glass-forming fluid to determine the most frequent structural arrangements~\cite{coslovich:2007,coslovich:2011}. This analysis recognizes that the locally preferred structure is not necessarily icosahedral or polytetrahedral, but it also correlates the extension of the preferred arrangement to the slowdown of relaxation. A related procedure has been put forward by Royall et al., with alternative characterizations of the local environment of the particles~\cite{malins:2012}. It should however be stressed that such work focuses on the frequency of a given preferred local arrangement, from which one can extract a length scale by dimensional analysis, rather than on the spatial correlations associated with the preferred local structure. By contrast, a measure of correlations has been attempted through the use of bond-orientational local order parameters, with varying degree of success~\cite{steinhardt:1983,sausset:2010,leocmach:2012,charbonneau:2012}.

From a spatial dimension point of view, the notion of a preferred cluster is hard to defend. The surface of a higher-dimensional hard sphere is so large that none of the possible configurations is likely to be significantly more represented than the others. But even if a class of them were to be more common, this order should be detected by a Delaunay simplicial decomposition or equivalently by its dual, the Voronoi decomposition.

\section{Geometrical Measures of fluid order in any Euclidean dimension}
\label{sect:liquidorder}
In order to determine the extent to which the fluid order in  $d$=3 may be singular we consider the development of the generalized spindle order from the Delaunay simplicial tessellation in high-dimensional systems. The focus is on Euclidean space, as a way to put $d$=3 systems in perspective. Although the decorrelation property suggests that all structural correlations in higher dimensions can be factorized down to the pair correlation function $g(r)$~\cite{torquato:2006}, going back and forth between $g(r)$ and the tessellation is non-trivial. Even the ideal gas geometry (Poisson limit) is not generally solvable. The tessellation also offers a differently intuitive description of the packing geometry than $g(r)$ alone. 

In order to consistently describe the tessellation order, we generalize the  $d=3$ relation established by Coxeter~\cite{coxeter:1961,Nelson:1989}, 
\begin{equation}\label{eqn:Zbardim3}
\bar{Z}_3=\frac{12}{6-\bar{q}_{1,3}}.
\end{equation}
When $d>3$, the analogous formula depends in part on the Euler characteristic $\chi$, which is always zero in odd dimension and which we take to be zero in even dimension as our systems evolve under periodic boundary conditions (the periodically flat torus). Similar formulae can then be derived in all dimensions following the general expression we establish in Theorem \ref{thm:Zbar} (Appendix~\ref{sect:appeuler}). An alternative and more numerically stable expression
\begin{equation}\label{eqn:Zbar2}
\bar Z_d=d\frac{\bar{q}_{0,d}}{\bar{q}_{1,d}}.
\end{equation}
is also available in all dimensions (Appendix~\ref{sect:appeuler}).
Using these general relations, we can obtain two analytically tractable limits for disordered packing: the Poisson, i.e., ideal gas, limit, and the ideal dense packing limit (Table~\ref{tbl:Zbar}). We consider these two cases before numerically evaluating the behavior of hard-sphere fluids at intermediate densities.

\subsection{Ideal packing limit}
\label{sect:idealpacking}
We first consider the densest possible configuration, i.e., the ideal packing, following   Coxeter's suggestion for the  statistical honeycomb in Ref.~\onlinecite[Sect 22.5]{coxeter:1961}. 
In $d$=3, one imagines the densest configuration obtained by the idealized version suggested by Eq.~(\ref{eqn:Zbardim3}).  One can wrap $\bar{q}_{1,3}^{\ideal}=\frac{2\pi}{\arccos(1/3)}$  regular polyhedra around an edge, and thus $\bar Z_3^\ideal \simeq 13.4$.  The calculation of $\bar{q}_{1,3}^{\ideal}$ is straightforward, and provides the ground for generalization:  the angle between two adjacent faces in a regular tetrahedron is $\arccos(\frac13)$.  This angle can be seen as the central angle of a triangle formed with the center of mass of an edge of the tetrahedron and the two vertices not on that edge. One then divides the total circumference of a circle of unit radius (the minimum distance between spheres) by the length of the arc cut by this triangle.

For $d>3$, consider a regular $d$-simplex $\sigma_d$ with unit side-length and vertices $\vec{p}_0,\dots,\vec{p}_d$. Pick a $i$-simplex $\sigma_i$ in it, say the one formed by the vertices $\vec{p}_0,\ldots,\vec{p}_i$.   The number $\bar{q}_{i,d}^{\ideal}$ is then obtained by dividing the volume of the $(d-i-1)$-dimensional sphere $S$  centered at the center of mass $\bar {\vec{p}}$ 
of $\sigma_i$ and passing through the remaining vertices by the volume of the section of the sphere cut by the non-regular $(d-i)$-simplex formed by the vertices $\bar {\vec{p}},\vec{p}_{i+1},\ldots,\vec{p}_d$.  This idea is illustrated in Fig~\ref{fig:CenterOfMass}.
		\begin{figure}
		\center{
	\includegraphics[width=0.5\columnwidth]{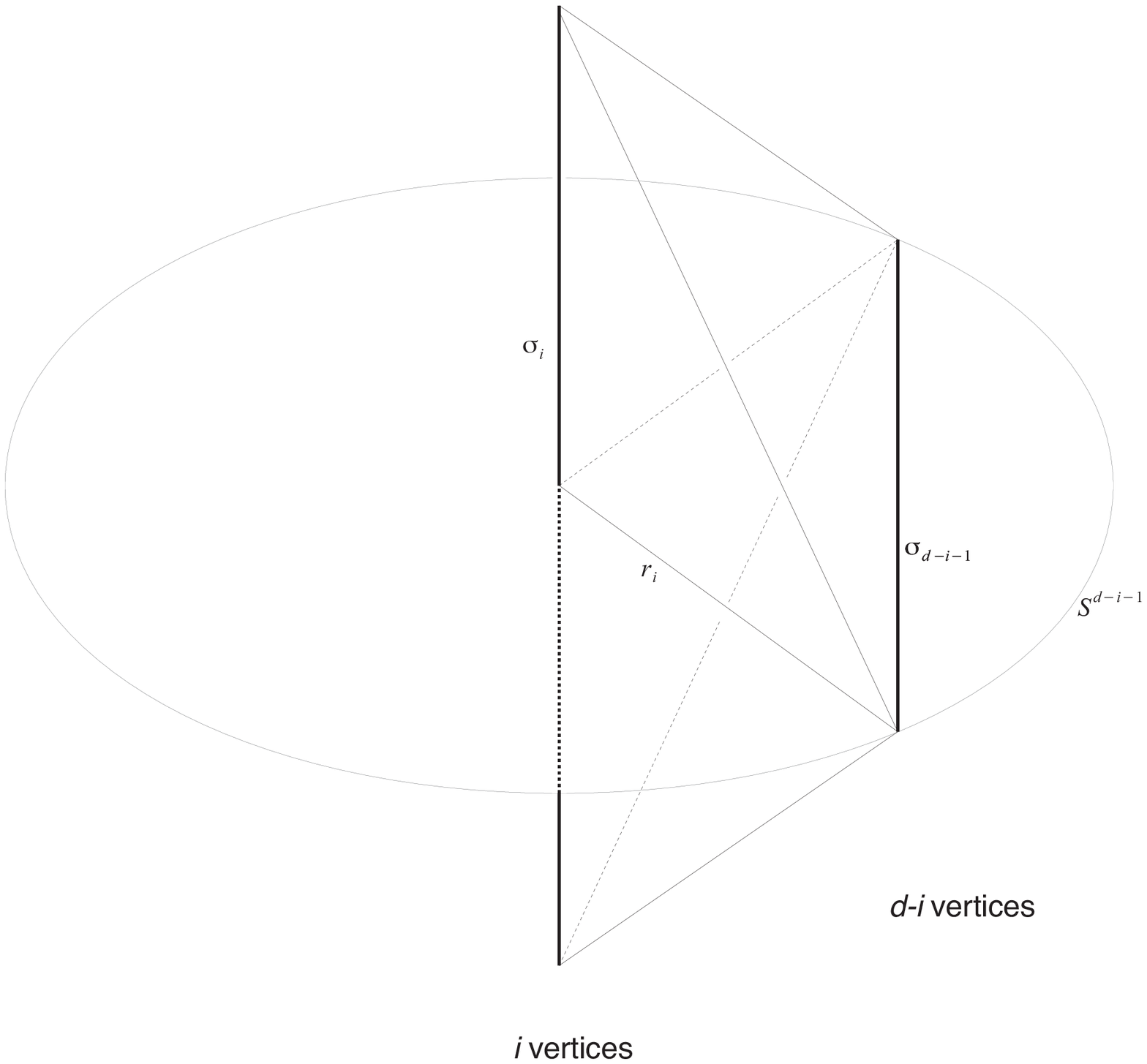}}
		\caption{Schematic of the integration setup for computing $\bar q_{i,d}^{\ideal}$}
		\label{fig:CenterOfMass}
		\end{figure}
The radius $r_i$ of the  sphere $S$ satisfies $r_i=\sqrt{\frac{i+2}{2i+2}}$ (Appendix \ref{app:ideal}).   It is convenient to rescale the problem and obtain $\bar{q}_{i,d}^{\ideal}$ as the quotient of the volume $\omega_{d-i-1}$ of a $(d-i-1)$-sphere of unit radius centered at $\vec{0}$ by the volume of the section of the sphere cut by the non-regular $(d-i)$-simplex formed by the vertices $\vec{0},\vec{r}_1,\ldots,\vec{r}_{d-i}$ with $|\vec{r}_j-\vec{r}_l|=\frac1{r_i}$ whenever $j\neq l$.  Let $\vec{v}_1,\ldots, \vec{v}_{d-i}$ be unit vectors perpendicular to the $(d-i-1)$-faces of this simplex touching  the origin and facing in the simplex.  We prove in Appendix~\ref{app:ideal} that $\vec{v}_i\cdot \vec{v}_j=-\frac1d$.  

Given $n$ vectors $\vec{v}_1,\ldots, \vec{v}_n\in\R^n$ such that 
\[\vec{v}_i\cdot \vec{v}_j=\begin{cases}k,&\text{ if } i\neq j,\\ 1,& \text{ if }i=j,\end{cases}\]
we define the region
\begin{align*}
S(k,n,p)&\equiv \{\vec{x}\in\R^n\mid \|\vec{x}\|=1, \text{ and } \vec{x}\cdot \vec{v}_i\geq p \text{ for all }i \},
\end{align*}
and let
$\tilde{V}(k,n,p)\equiv\Vol(S(k,n,p))$.
We have established that
\begin{equation}\label{eqn:qideal}
	\bar{q}_{i,d}^{\ideal}=\frac{\omega_{d-i-1}}{\tilde{V}(-\frac1d,d-i,0)}.
\end{equation}

There is a recursive formula for computing $\tilde{V}(k,n,p)$, as demonstrated in Appendix~\ref{app:ideal}:
\begin{theorem}\label{prop:iteration}
Let \[f(k,p)\equiv k-2p^2k-2p\sqrt{1-p^2}\sqrt{1-k^2},\text{ and}\]
\begin{widetext}
\begin{equation*}
	h(k,n,p)\equiv\frac{(n-1)pk+\sqrt{(1-k)\bigl((n-1)k+1\bigr)\bigl(1-(n-1)p^2+(n-2)k\bigr)}}{(n-2)k+1}.
\end{equation*}
We have
\begin{equation}
	\tilde{V}(k,n,p)=\begin{cases}
	\arccos(-k), & \text{ if $n=2$ and $p=0$,}\\
	\tilde{V}(f(k,p),2,0), & \text{ if $n=2$ and $p\neq 0$,}\\
\displaystyle	\int_p^{h(k,n,p)} (1-z^2)^{\frac{n-3}2}\tilde{V}\left(\frac{k}{1+k},n-1,\frac{p-kz}{\sqrt{1-k^2}\sqrt{1-z^2}}\right)dz, &\text{ otherwise.}\end{cases}
\end{equation}\end{widetext}
\end{theorem}

Using this expression, we immediately obtain that
\begin{equation}
	\bar{q}_{d-2,d}^{\ideal}=\frac{2\pi}{\arccos(\frac1d)}.
\end{equation}

One can also compute the packing occupied by that fictitious construction, and obtain, as was done in Ref.~\onlinecite{Nelson:1989} for $d=3$, Roger's bound~\cite{rogers:1958}. Let $2\sigma_d$ be a regular simplex in $\R^d$ of side length $2$ having one vertex at the origin in $\R^d$.  Let  $B^d$ be the ball of radius $1$ centered at the origin and $S^{d-1}$ be its boundary, the sphere in $\R^d$.  The packing fraction (or density) is then obtained by 
\[\varphi_d=\frac{(d+1)\Vol(2\sigma_d\cap B^d)}{\Vol(2\sigma_d)}.\]
In other words, $\varphi_d$ is the ratio of the volume of the part of the simplex covered by the unit spheres centered at its $(d+1)$ vertices to the volume of the whole simplex. 

The packing fraction is thus
\begin{align*}
	\varphi_d^\ideal &=(d+1)\frac{\tilde V(-\frac1d,d,0)\frac{\mathrm{Vol}(B^d)}{\mathrm{Vol}(S^{d-1})}}{2^d\frac{\sqrt{d+1}}{d!\sqrt{2^d}}}\\
	&=\frac{(d-1)!\sqrt{d+1}}{\sqrt{2^d}}\tilde V(-\frac 1d,d,0),
\end{align*}
which is computed for $d$=3--8 in Table \ref{table:densities}. Note that in Ref.~\onlinecite[p.15]{conway:1988}, one has that same bound $\delta_d^\mathrm{bound}$ on the center density $\delta$, with  $\varphi_d=\frac{\pi^{d/2}\delta_d^\mathrm{bound}}{(\frac d2)!}$.

\begin{table*}
	\centering
\caption{Results for $\bar Z_d$. Formulae A and B are formally equivalent, but the former is much more numerically unstable than the latter as dimension increases.}\label{tbl:Zbar}
\begin{tabular}{|c|c|c|c|c|}
\hline 
Stat&Formula A  &Formula B & Ideal packing (densest) & Poisson (ideal gas)\\ \hline\hline
$\bar{Z}_3$&$\frac{12}{6-\bar{q}_{1,3}}$ &$\displaystyle 3\bar{q}_{0,3}/\bar{q}_{1,3}$
& $\frac{12}{6-2\pi\arccos(\frac13)}\simeq 13.397$   &$2+\frac{48\pi^2}{35}\simeq 15.536$\\ \hline
$\bar{Z}_4$&$\frac{20}{10-\bar{q}_{1,4}\left(\frac{10}{\bar{q_{2,4}}}-\frac{3}{2}\right)}$ &$\displaystyle 4\bar{q}_{0,4}/\bar{q}_{1,4}$
&26.44(3)&$\frac{1430}9\simeq 37.778$\\ \hline
$\bar{Z}_5$&$\frac{30}{ 15- \bar{q}_{1,5}\left( \frac{20}{\bar{q}_{2,5}}-\frac{15}{\bar{q}_{3,5}}
+2\right) }
$ &$\displaystyle 5\bar{q}_{0,5}/\bar{q}_{1,5}$
&$48.68(7)$&$88.4(8)$  \\ \hline
$\bar{Z}_6$&$\frac{42}{21-\bar q_{{1,6}} \left( \frac{35}{\bar q_{{2,6}}}-\frac{35}{\bar q_{{3,6}}}+\frac{21}{\bar q_{{4,6}}}-\frac{5}2 \right)  }
$ &$\displaystyle 6\bar{q}_{0,6}/\bar{q}_{1,6}$
&$85.5(4)$&$202.5(9)$ \\ \hline
$\bar{Z}_7$&$\frac{56}{28-\bar{q}_{1,7}\left(\frac{56}{\bar{q}_{2,7}}-\frac{70}{\bar{q}_{3,7}}+\frac{56}{\bar{q}_{4,7}}-\frac{28}{\bar{q}_{5,7}}+3\right)}$ &$\displaystyle 7\bar{q}_{0,7}/\bar{q}_{1,7}$
&$145(2)$&  $458(4)$\\ \hline
$\bar{Z}_8$&$\frac{72}{36-\bar q_{{1,8}} \left( \frac{84}{\bar q_{{2,8}}}-\frac{126}{\bar q_{{3,8}}}+\frac{126}{\bar q_{{4,8}}}-\frac{84}{\bar q_{{5,8}}}+\frac{36}{\bar q_{{6,8}}}-\frac72 \right)}$ &$\displaystyle 8\bar{q}_{0,8}/\bar{q}_{1,8}$
&
$243(3)$ 
&$102(2)\times 10$ \\ \hline \hline
$\bar{Z}_d$&$	\frac{d(d+1)}{\binom{d+1}{2}-\bar{q}_{1,d}\left((-1)^{d-1}\frac{d-1}2+\sum_{j=2}^{d-2} \frac{(-1)^j\binom{d+1}{j+1}}{\bar{q}_{j,d}}\right)}	
$ 
&$d\frac{\bar{q}_{0,d}}{\bar{q}_{1,d}}$& \multicolumn{2}{c|}{See appendix for details}\\ \hline
\end{tabular}
\end{table*}

\begin{table}
	\centering
\caption{Ideal packing fraction $\varphi_d$ in various dimensions}\label{table:densities}
\begin{tabular}{|l||l|}
	\hline
	$\varphi_3^\ideal=\sqrt{2}\bigl(3\arccos(\frac13)-\pi\bigr)\simeq 0.7796$\quad &$\varphi_6^\ideal\simeq 0.423(5)$\\
	$\varphi_4^\ideal\simeq 0.647(1)$      &$\varphi_7^\ideal\simeq 0.33(2)$\\
	$\varphi_5^\ideal\simeq 0.5256(1)$      &$\varphi_8^\ideal\simeq 0.26(1)$\\
	\hline
\end{tabular}
\end{table}

\subsection{Poisson Limit}
We now examine the sparsest configuration possible, i.e., the ideal gas, which we model by a Poisson process whose geometrical properties self-average in the limit of large $N$.  Some of the statistics in which we are interested can be found in the encyclopedic Ref.~\onlinecite{okabe:2000}.  For our purpose, we use the results from Ref.~\onlinecite{SchneiderWeilBook} and hence their notation. Note, however, that we refer to a ``tessellation'' in lieu of their ``mosaic.''

Given $\tilde X$ a random Poisson process in $\R^d$ with intensity $\gamma$, one can find its Voronoi decomposition $X$ and Delaunay decomposition $Y$.  Both are also random processes with their own intensities, and their various skeletons all have their own intensities.  Given a tessellation $Z$, its $j$-skeleton $Z^{(j)}$ is the set of all its faces of dimension $j$.  Let $\beta^{(j)}$ be the intensity of the random process $Y^{(j)}$ and $\gamma^{(j)}$ be the intensity of the random process $X^{(j)}$.  Because the points in $\tilde X$ are vertices of the Delaunay decomposition,  $\beta^{(0)}=\gamma$, and because the Voronoi and Delaunay decompositions are dual to each other, each vertex of $Y$ corresponds to a $d$-dimensional face of $X$, i.e., $\gamma^{(d)}=\gamma$.  More generally, $\beta^{(j)}=\gamma^{(d-j)}$, from Ref.~\onlinecite[Thm 10.2.8]{SchneiderWeilBook}.

We can now extend the notion of coordination number we defined above.  For $F\in Z^{(j)}$,  let 
\[N_{j,k}(F)=\# \{S\in Z^{(k)}\mid S\cap Z\neq \emptyset\}\]
and let
\[n_{j,k}(Z)\equiv\EE(N_{j,k})\]
be the average value, so $n_{j,d}=\bar{q}_{j,d}$.

For any stationary random tessellation (the Poisson--Voronoi and Poisson--Delaunay tessellations are both stationary) of face intensities $\lambda^{(k)}$, we have  (Ref.~\onlinecite[Thm 10.1.2]{SchneiderWeilBook})
\begin{equation}
	\lambda^{(j)}n_{j,k}=\lambda^{(k)}n_{k,j}.\label{eqn:swap}
\end{equation}
With certainty, the Poisson--Delaunay tessellation is simplicial.  In a simplicial decomposition, each $k$-face is a $k$-simplex and therefore each contains $\binom{k+1}{j+1}$ $j$-faces when $j\leq k$.  Hence $n_{k,j}(Y)=\binom{k+1}{j+1}$ when $j\leq k$.  It turns out that  a dual phenomenon occur for the Poisson--Voronoi tessellation: every $k$-face is contained in precisely $d-k+1$ top dimension cell.  Hence $n_{k,d}(X)=d-k+1$ and $n_{j,k}(X)=\binom{d-j+1}{k-j}$ for $0\leq j\leq k\leq d$. We thus obtain
\begin{align*}
\bar q_{k,d}^{\Poisson}&=n_{k,d}(Y)=\frac{\beta^{(d)}}{\beta^{(k)}}n_{d,k}(Y)=\frac{\gamma^{(0)}}{\gamma^{(d-k)}}\binom{d+1}{k+1}\\
	&=\frac{n_{(d-k),0}(X)}{n_{0,(d-k)}(X)}\binom{d+1}{k+1}=\frac{\binom{d+1}{k+1}}{\binom{d+1}{d-k}}n_{(d-k),0}(X)\\
	&=n_{(d-k),0}(X),
\end{align*}
and
\begin{align*}
	\bar Z_d^\Poisson&=n_{0,1}(Y)=n_{1,0}(Y)\frac{\beta^{(1)}}{\beta^{(0)}}=2\frac{\gamma^{(d-1)}}{\gamma^{(d)}}\\
	&=2\frac{n_{d,d-1}(X)}{n_{d-1,d}(X)}=n_{d,d-1}(X).
\end{align*}

To find the various $n_{j0}(X)$, we use the relation given by Eq.~(\ref{eqn:swap}), the relations
\begin{align*}
   \sum_{k=0}^j (-1)^kn_{j,k}=1,\quad\quad 	\text{\cc{Eqn (10.14)}},\\
\sum_{k=j}^d (-1)^{d-k}n_{j,k}=1,\quad\quad 	\text{\cc{Eqn (10.17)}},   \\
\sum_{i=0}^d\gamma^{(i)}=0,\quad\quad\text{\cc{Eqn (10.21)}},
\end{align*}
valid for any stationary random tessellation (including the Poisson--Delaunay), the relations
\begin{gather*}
	n_{j,k}=\binom{d-j+1}{k-j},\text{ when $j\leq k$},\quad\text{\cc{p.448}},\\
	(1-(-1)^k)\gamma^{(k)}=\sum_{j=0}^{k-1}(-1)^j
\binom{d+1-j}{k-j}\gamma^{(j)},
\end{gather*}
from Ref.~\text{\cc{Thm 10.1.5}}, valid for the Poisson--Voronoi tessellation $X$, and the explicit expression
\[\gamma^{(0)}=\frac{2^{d+1}\pi^{\frac{d}2}}{d(d+1)!}\frac{\Gamma(\frac{d^2+1}2)\Gamma(1+\frac d2)^{d}\Gamma(d)}{\Gamma(\frac{d^2}{2})\Gamma(\frac{d+1}2)^{d}\Gamma(\frac{1}2)}\gamma\]
for the intensity of the vertices of the Poisson--Voronoi tessellation (see \cc{Thm 10.2.4}). These equations contain enough information to fully determine the intensities $\gamma^{(0)},\ldots, \gamma^{(d)}$ and, in turn, all the $n_{jk}(X)$ when $d=3$ (see \cc{Thm 10.2.5} and Table \ref{tbl:stats}) and when $d=4$.  We can indeed compute

\begin{theorem}Let $X$ be a Poisson--Voronoi tessellation of intensity $\gamma$ in $\R^4$.  Then
\begin{align*}
\gamma^{(0)} = \frac{286}9\gamma,\gamma^{(1)} = \frac{715}9\gamma, \gamma^{(2)} = \frac{590}9\gamma,\gamma^{(3)} = \frac{170}9\gamma,
\end{align*}
\begin{align*}
\bar{q}_{2,4}^\Poisson&={\frac {286}{59}},& \bar{q}_{0,4}^\Poisson&={\frac {1430}{9}},\\
\bar{q}_{1,4}^\Poisson&={\frac {286}{17}},& \bar{Z}_4^\Poisson&={\frac {340}{9}},
\end{align*}
\begin{gather*}
n_{{2,1}}(X)={\frac {286
}{59}},
n_{{3,1}}(X)={\frac {429}{17}},n_{{4,1}}(X)={\frac {
2860}{9}},\\
n_{{3,2
}}(X)={\frac {177}{17}},n_{{4,2}}(X)={\frac {590}{3}}.
\end{gather*}
\end{theorem}

Although this theorem is not explicitly written in Ref.~\onlinecite{SchneiderWeilBook}, it was clearly accessible to its authors. We include it here for the reader's benefit.

When $d=5$ and $d=6$, the same method allows one to find the intensities $\gamma^{(0)}$ and $\gamma^{(1)}$. The intensities $\gamma^{(2)},\ldots,\gamma^{(d-1)}$ are unknown, but are all dependent on $\gamma^{(2)}$.  Given that $\bar{q}_{d-2,d}^\Poisson$ depends in the missing intensity, we can use the simulation results to find the missing factor (see Appendix~\ref{app:Poisson}), and in turn approximate the values of the other quantities. A similar situation occurs in $d=7$ and $d=8$, but two intensity parameters are then needed. See Table \ref{tbl:stats} for results and Appendix \ref{app:Poisson} for details.

\begin{table}
	\centering
\caption{Limit coordination numbers. }\label{tbl:stats}
\begin{tabular}{|c|c|c|}
\hline 
Stat& Ideal packing (densest) & Poisson (ideal gas)\\ \hline
$\bar{q}_{1,3}$&$\frac{2\pi}{\arccos(\frac13)}\simeq 5.104$&$\displaystyle\frac {144\pi^2}{24\pi^2+35}\simeq 5.228$\\ \hline
$\bar{q}_{0,3}$&$\frac{4\pi}{3\arccos(\frac13)-\pi}\simeq22.795$&$\displaystyle\frac {96\pi^2}{35}\simeq 27.071$\\ \hline
\hline 
$\bar{q}_{2,4}$&$\frac{2\pi}{\arccos(\frac14)}\simeq 4.767$&$\displaystyle\frac {286}{59}\simeq 4.847$\\ \hline
$\bar{q}_{1,4}$&$15.461(1)$&$\displaystyle\frac {286}{17}\simeq 16.824$\\ \hline
$\bar{q}_{0,4}$&$102.3(5)$&$\displaystyle\frac {1430}{9}\simeq 158.889$\\ \hline
\hline 
$\bar{q}_{3,5}$&$\frac{2\pi}{\arccos(\frac15)}\simeq 4.589$&$\displaystyle 4.642(2)$\\ \hline
$\bar{q}_{2,5}$&$12.999(1)$&$\displaystyle  13.673(4)$\\ \hline
$\bar{q}_{1,5}$&$53.414(1)$&$\displaystyle  63.36(4)$\\ \hline
$\bar{q}_{0,5}$&$520(2)$ &$\displaystyle \frac{7776000}{676039}\pi^4\simeq 1120.42$\\ \hline
\hline 
$\bar{q}_{4,6}$&$\frac{2\pi}{\arccos(\frac16)}\simeq4.477$&$\displaystyle 4.51(1)$\\ \hline
$\bar{q}_{3,6}$&$11.761(1)$&$\displaystyle 12.1572(8)$\\ \hline
$\bar{q}_{2,6}$&$39.527(1)$&$\displaystyle  43.494(8)$\\ \hline
$\bar{q}_{1,6}$&$205.0(4)$&$\displaystyle  268.9(5)$\\ \hline
$\bar{q}_{0,6}$&$2.92(4)\times10^{3}$&$\displaystyle \frac{90751353}{10000}= 9075.1353$\\ \hline
\hline 
$\bar{q}_{5,7}$&$\frac{2\pi}{\arccos(\frac17)}\simeq4.4017$&$\displaystyle 4.4292(1)$\\ \hline
$\bar{q}_{4,7}$&$11.016(1)$&$\displaystyle 11.279(1)$\\ \hline
$\bar{q}_{3,7}$&$33.1(1)$&$\displaystyle  35.23(3)$\\ \hline
$\bar{q}_{2,7}$&$131.9(3)$&$\displaystyle  153.09(1)$\\ \hline
$\bar{q}_{1,7}$&$862(3)$&$\displaystyle  1.26(2)\times 10^3$\\ \hline
$\bar{q}_{0,7}$&$1.79(3)\times 10^4$&$\displaystyle \frac{275365800000}{322476036831}\pi^6\simeq 82094.097$\\ \hline
\hline 
$\bar{q}_{6,8}$&$\frac{2\pi}{\arccos(\frac18)}\simeq4.3467$&$4.3679(2)$\\ \hline
$\bar{q}_{5,8}$&$10.517(1)$&$10.704(1)$\\ \hline
$\bar{q}_{4,8}$&$29.394(1)$&$30.73(1)$\\ \hline
$\bar{q}_{3,8}$&$100.9(3)$&$111.2(1)$\\ \hline
$\bar{q}_{2,8}$&$472(3)$&$586.2(1)$\\ \hline
$\bar{q}_{1,8}$&$3.92(2)\times 10^3$&$642(3)\times 10$\\ \hline
$\bar{q}_{0,8}$&$1.19(5)\times 10^5$&$\displaystyle \frac{37400492672297766}{45956640625}\simeq 813821$\\ \hline
\end{tabular}
\end{table}

\subsection{Simulation Approach}
\begin{figure}
\includegraphics[width=\columnwidth]{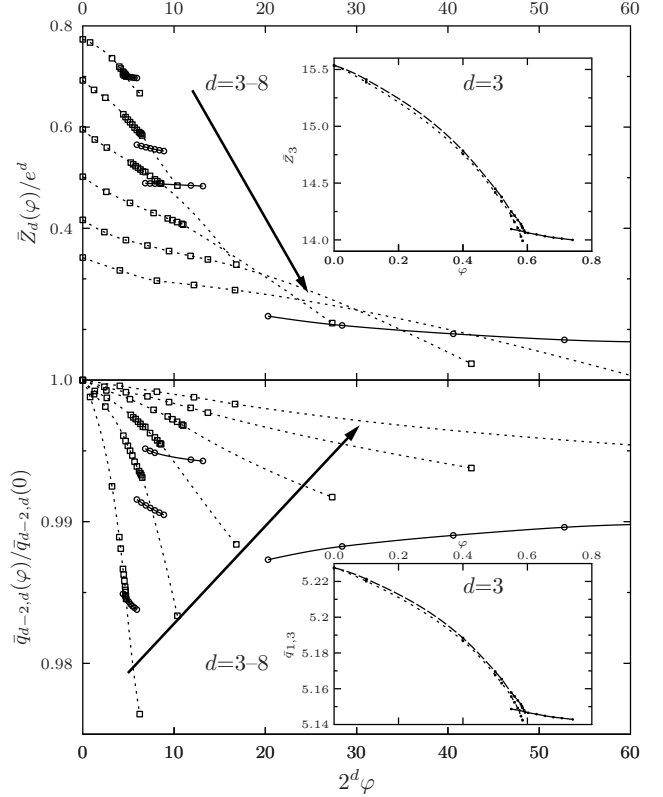}
\caption{Density evolution of the average bond spindle $\bar{q}_{d-2,d}$ and number of nearest neighbors $\bar{Z}_d$ with density for fluid (squares, dashed line) and crystal phases (circle, solid line). The fluid behavior continuously evolves from the Poisson limit and reasonably extends to the ideal packing limit (rightmost point). The crystal phase, however, clearly sits on a different branch and even $E_8$ does not converge to the ideal limit. (Inset) Results for the 7:5 (long-dashed line) and 6:5 (short-dashed line) mixtures in $d$=3, along with the crystal results for a monodisperse spheres (solid line).}
\label{fig:highdsimplex}
\end{figure}

\begin{figure}
\includegraphics[width=\columnwidth]{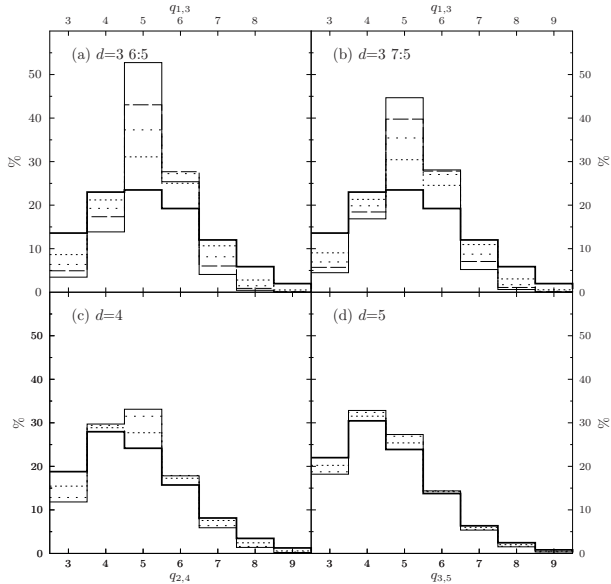}
\caption{Density evolution of the bond spindle distribution for the Poisson limit (thick line), the highest equilibrated fluid density attained (solid line -- see Fig.~\ref{fig:HSfragility}), and a couple of intermediate densities. The narrower distribution of the 6:5 mixture compared to that of the 7:5 mixture at similarly sluggish conditions indicates that the former is closer to the ideal packing order than the latter. With increasing dimension, however, the distribution converges to that of the Poisson limit. }
\label{fig:highdspindle}
\end{figure}

Equilibrated hard-sphere fluid configurations with at least $N=8000$ are obtained in $d$=3--8, under periodic boundary conditions, using a modified version of the event-driven molecular dynamics code described in Refs.~\onlinecite{skoge:2006,charbonneau:2012b}.  Time is expressed in units of $\sqrt{\beta m\sigma^2}$ for particles of unit mass $m$ at fixed unit inverse temperature $\beta$. We also consider finite-pressure, hard-sphere crystals of $D_3$, $D_4$, $D_5$, and $E_8$ lattice symmetry, which form the densest known packings in $d$=3, 4, 5, and 8, respectively.  Monodisperse spheres are used, except for $d$=3 fluids, where 50\%:50\% binary mixtures of diameter ratio 7:5 and 6:5 (the larger sphere diameter sets the unit length) are used to prevent crystallization, while systematically approaching the monodisperse limit. The densest packings of these systems is indeed thought to be fractionated~\cite{hopkins:2012}, and the complex phase behavior of similar mixtures at finite pressure suggests that the drive to crystallize is very small~\cite{kranendonk:1991}. The properties of these mixtures have also been extensively characterized, e.g., Refs.~\onlinecite{berthier:2009,flenner:2011}, and Refs.~\onlinecite{foffi:2003,foffi:2004}, respectively.  

The diffusivity $D$ is obtained by measuring the long-time behavior of the mean-squared displacement $\lim_{t\rightarrow\infty}\langle (\Delta r)^2\rangle=\frac{1}{N}\sum_i [\vec{r}_i(t)-\vec{r}_i(0)]^2=2dDt$, while the pressure $P$ is mechanically extracted from the collision statistics. The QuickHull package~\cite{barber:1996} provides a Delaunay simplicial tessellation, which is analyzed to extract $\bar{q}_{i,d}$ and $\bar{Z}_d$. In $d<6$, a tessellation can be directly obtained for full periodic configurations, but in $d\geq 6$, the tessellation can only be obtained particle by particle, in order to minimize memory usage. In the second approach, the values of $\bar{q}_{j,d}$ are obtained by averaging the results over the different particles. In both cases, the Euler characteristic is explicitly computed to check the consistency of the approach. Due to computational limitations, the results in $d$=8 are only averaged over about 100 particles.

\subsection{Simulation Results}
\begin{figure}
\center{\includegraphics[width=\columnwidth]{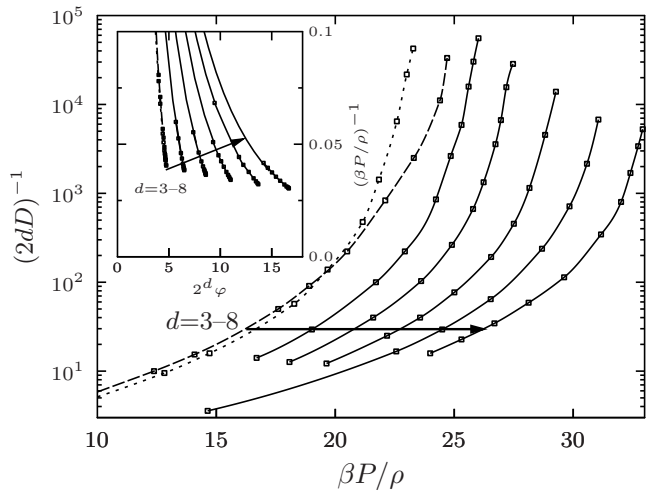}}
\caption{The diffusivity decays with increasing reduced pressure $\beta P/\rho$ in a similar fashion for all dimensions. The behavior of $d$=3 6:5 (short-dashed line) and 7:5 mixtures (long-dashed line) are also similar, although the former initiates its slowdown at lower pressures than the latter. (Inset) The equation of state for the different fluids. On this scale the results for the two $d$=3 mixtures are indistinguishable.}
\label{fig:HSfragility}
\end{figure}

As expected, the values of $\bar{q}_{j,d}$ and $\bar{Z}_d$ are found to be intermediate between the Poisson and the ideal packing limits  for all \emph{fluid} configurations (Fig.~\ref{fig:highdsimplex}). The fluid structure evolves smoothly from the Poisson to the ideal packing limit, but without ever actually reaching the latter. In $d$=3, the 6:5 mixture approaches it more than the 7:5, in agreement with the notion that the former is less geometrically frustrated than the latter. It also allows one to reasonably extrapolate that the a non-crystallizing monodisperse system would have a qualitatively similar behavior that what is found for the mixtures.  The fluid order thus systematically tends toward an ideal simplex, or generalized polytetrahedral, packing.

As $d$ increases, however, the ideal packing limit becomes ever more distant from the dynamically sluggish regime~\cite{charbonneau:2012b}. Not only the mean, but the entire bond spindle distribution steadily converges to the Poisson limit (Fig.~\ref{fig:highdspindle}). In that sense, the highest-density equilibrated fluid that is dynamically accessible (keeping diffusivity constant) gets structurally ever closer to the Poisson limit, rather than to an ideal simplex packing as $d$ increases.  This interpretation is consistent with higher dimensions corresponding to frustrated spaces, because a relatively high curvature would be needed to obtain a tiling of simplices. Although not a measure of structural correlation directly, it suggests that extracting such a length from geometry would necessarily shrink with $d$. Because the static correlation length necessary for generating an equivalent slowing down of the dynamics shrinks with dimension, this inference is not in contradiction with static order underlying the process. We note, however, that the order in $d=3$ systems is \emph{not} singularly closer to the ideal packing limit, even though the ideal simplex packing is much closer to Euclidean space than for $d>3$. Remarkably, the dynamical slowdown in Figure~\ref{fig:HSfragility} also evolves smoothly with $d$. This continuum of behavior suggests that any explanation that provides a singular role to the geometry of $\R^3$ cannot reasonably hold. 

Another question of interest concerns the relationship between the fluid and crystal branches. For perfect crystal packings, the crystal's Delaunay tessellation is not generally simplicial. Other regular polytopes can be identified. Yet when adding thermal noise, and even in the limit of infinitely small noise, simplices are recovered with probability one. As a result, in $d$=3 the face-centered-cubic crystal has $\bar{Z}_3$=14, not 12, when approaching crystal close packing from lower densities~\cite{troadec:1998}. In higher-dimensional crystals, a similar behavior is observed, but no prediction for the limit behavior of $\bar{Z}_d$ is known. In general, the crystal branch is clearly different from the fluid branch,  confirming that no crystallite form in these systems, in contrast to what happens in $d$=3 monodisperse hard spheres where crystallization is facile~\cite{anikeenko:2007,anikeenko:2008}. Even in $d$=8, where the $E_8$ crystal contains an inordinately high fraction of simplices, the fluid order follows a branch that is clearly distinct from the crystal branch. The crystal bond spindle behavior is even located outside of the fluid bounds. The dense $d$=8 fluid thus bears no trace of the high-simplicial nature of the crystalline phase. Actually, even the crystal structure is but partially captured by the ideal packing.  We stress that no MRCO is found in high-$d$ fluids, confirming the conclusions of earlier fluid order studies~\cite{vanmeel:2009b}.

\section{``Order agnostic'' lengths}
\label{sect:liquidagnostic}
The results of the above analysis indicate that the simple geometrical considerations brought forward to explain the dynamical slowdown are unlikely to provide any significantly growing static correlation length in hard-sphere glass-formers under compression. Other measures of such a correlation length, \textit{e.g.} through bond-orientational order, confirm this conclusion (see Ref.~\onlinecite{charbonneau:2012} for the case $d=3$). Yet this analysis does not exhaust the possibilities for more subtle forms of structural correlations. Other proposals have been made based on nontrivial ``order agnostic'' lengths~\cite{sausset:2011}, such as the patch repetition length associated with the notion of configurational entropy~\cite{kurchan:2011,sausset:2011}, length scales extracted from finite-size~\cite{karmakar:2009,berthier:2012b} and information theoretic analysis~\cite{ronhovde:2011}, as well as point-to-set correlation lengths~\cite{bouchaud:2004,montanari:2006,biroli:2008,berthier:2012,berthier:2012b,hocky:2012,kob:2011}.

The point-to-set correlation lengths, in particular, can be studied by considering the distance over which boundary conditions imposed by pinning particles in a fluid configuration affect the equilibrium structure of the remaining (unpinned) particles. The original proposal, motivated by the random first-order transition theory~\cite{kirkpatrick:1989}, considered a cavity whose exterior is a frozen fluid configuration~\cite{bouchaud:2004}. The crossover length characterizing the range over which the pinned boundary determines the configuration inside such a cavity also enters in a formula bounding from above the relaxation time of the (bulk) fluid~\cite{montanari:2006}. Yet other geometries of the set of pinned particles also allow one to extract point-to-set correlation lengths~\cite{berthier:2012}. These lengths need not coincide nor evolve in exactly the same way as the system becomes sluggish, but they nonetheless characterize the same underlying physics. Evidence for the growth of such lengths as the relaxation time increases has been found in several model glass-formers~\cite{biroli:2008,berthier:2012,hocky:2012}.

These point-to-set correlation lengths are more general than structural lengths based on a specific local order description and are thus expected to provide upper bounds for the latter. Although they may provide less detailed geometrical information, and may be numerically harder to investigate over a broad dynamical range, it seems clear that a correlation between structure and dynamics, if present, should at least emerge at their level. For this reason, we complete our study of structural lengths related to ideal simplex order in hard-sphere glass-forming fluids with that of the ``order-agnostic'' point-to-set correlation length obtained from a random pinning of a fraction of the particles in equilibrium configurations.


\subsection{Overlap Function}
We consider a situation in which a fraction $c$ of the particles of an equilibrium fluid configuration is pinned at random. Information about a possible static length scale growing with the dynamical slowdown as one compresses the fluid is then obtained from the long-time limit of the overlap  between the original configuration and the configuration equilibrated in the presence of the pinned particles. If the reference and the final configurations are quite similar, the average pinning spacing is shorter than the static correlation length, and the opposite if the two configurations are dissimilar.
Previous studies have chosen an \emph{ad hoc} boxing of space~\cite{berthier:2012,charbonneau:2012,biroli:2008,hocky:2012}, but in systems where the particle density is changed, it is more robust to measure the configuration overlap by means of a microscopic overlap function 
\begin{equation}
w_{mn}(t)\equiv \Theta(a-|\mathbf{r}_n(t)-\mathbf{r}_m(0)|),
\end{equation}
where $a=0.3\sigma$ is chosen sufficiently small to enforce single overlap occupancy for hard spheres. For a fraction $c$ of pinned particles we therefore have
\begin{equation}
\label{eq_def_overlap}
Q_c(t)\equiv\frac{1}{(1-c)^2N}\left\langle \overline{\sum_{m,n \notin \mathcal B } w_{mn}(t)}\right\rangle,
\end{equation}
where the brackets denote an average over equilibrium configurations, the overline represents an average over the different ways to pin a fraction $c$ of the particles of a given equilibrium configuration, and the sum is over all unpinned particles, with $\mathcal B$ denoting the set of pinned particles.

When $t=0$, pinning is irrelevant and the overlap $Q_c(0)=1$. When there are no pinned particles, \textit{i.e.}, for $c=0$,
\begin{equation}
Q_0(t)=\frac{1}{N} \int d \mathbf{ r}\int d \mathbf{ r'} \, \Theta(a-|\mathbf{r}-\mathbf{r'}|) \left\langle \rho(\mathbf r,0)\rho(\mathbf{r',t}) \right\rangle,
\end{equation}
where $\rho(\mathbf r,t)=\sum_{n=1}^N \delta(\mathbf r - \mathbf{r}_n(t))$ is the density of particles in a configuration $\{\mathbf{r}_i(t) \}$ at time $t$. The overlap $Q_0(t)$ could be expressed in terms of van Hove functions whose weight is provided by the microscopic overlap function, but this approach would not be not particularly illuminating. Instead, at $t=0$, we use the equilibrium property~\cite{hansen:1986}
\begin{equation}
\label{eq_pairdensity}
\left\langle \rho(\mathbf r)\rho(\mathbf r')\right\rangle = \rho\, \delta(\mathbf r - \mathbf{r'}) + \rho^2 [1 + \rho^2 h(|\mathbf{r}-\mathbf{r'}|)] ,
\end{equation}
where $h( r)= g( r)-1$ is the total correlation function, together with the fact that for hard spheres $h( r<\sigma)=-1$  (and thus also when particles are within a distance $a<\sigma$), to recover that $Q_0(0)=1$. When $t\rightarrow \infty$, one finds
\begin{equation}
\label{eq_Q0_infty}
Q_0(\infty)=\frac{1}{N} \int d \mathbf{ r}\int d\mathbf{ r'}\, \Theta(a-|\mathbf{r}-\mathbf{r'}|) \rho^2 = \rho V_d(a) .
\end{equation}
The relevant quantity  one wishes to examine is thus $Q_c(\infty)-Q_0(\infty)$. By locating the rapid growth of the overlap with $(\rho c)^{-1/d}$, we can extract the crossover length between small and large overlap, which corresponds to a static point-to-set correlation length within the bulk system.

\subsection{Simulation approach and results}
\begin{figure}
\includegraphics[width=\columnwidth]{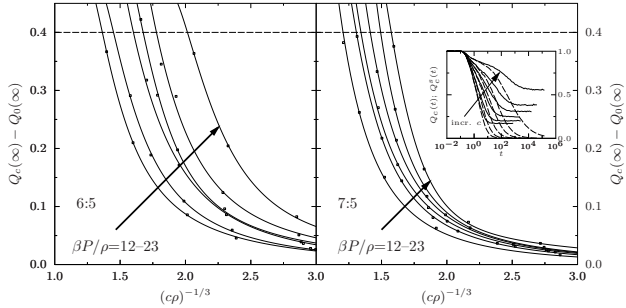}
\caption{Growth of the overlap with increasing pinning concentration in $d=3$ for two different hard-sphere mixtures. As the system density increases, the crossover from low overlap to high overlap takes place at an ever smaller concentration of pinned particles. The average spacing between defects at a 0.4 overlap defines $\xi_{\mathrm{p}}$. (Inset) Time evolution of the overlap and of the self component $Q_c^s$. The long time value of the overlap is attained when the self-part has completely decayed.}
\label{fig:newoverlap}
\end{figure}

\begin{figure}
\includegraphics[width=\columnwidth]{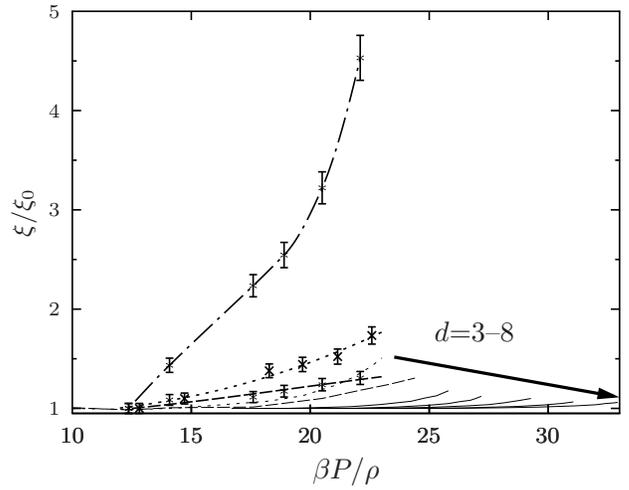}
\caption{The static (point-to-set) length $\xi_{\mathrm{p}}$ for $d$=3 6:5 (short-dashed line) and 7:5 (long-dashed line), extracted from the data in Fig.~\ref{fig:newoverlap}, both saturate the bound from Eq.~(\ref{eq_bound-length}) (thinner lines). In higher dimensions, the bound becomes continuously weaker. The $\xi_{\mathrm{dyn}}$ results for the 7:5 mixture (dot-dashed line) are taken from Ref.~\onlinecite{flenner:2011}.}
\label{fig:dynlength}
\end{figure}

We have calculated the point-to-set static length from the pinning analysis in $d$=3 only. In higher dimensions, the computation becomes prohibitively time consuming. Also, smaller systems with $N=1236$ are used and the absence of finite-size effects are verified for select systems with $N=9888$.  The length $\xi_\mathrm{p}$ is defined as the value of the mean distance $(\rho c)^{-1/d}$ between pinned particles for which the overlap falls below $0.4$ (Fig.~\ref{fig:newoverlap}). The extracted length is not very sensitive to this choice, provided it is intermediate between low and high overlap.

Figure~\ref{fig:dynlength} shows that the static length $\xi_\mathrm{p}$ increases only very modestly over the dynamically accessible 
density range, while the diffusivity $D$ and the structural relaxation time $\tau_\alpha$ change by about 4 orders of 
magnitude (Fig.~\ref{fig:HSfragility}). Interestingly, over the same density range, the ``dynamical length'' $\xi_\mathrm{dyn}$ 
characterizing the extent of dynamical heterogeneity 
grows markedly (the results shown for the 7:5 mixture are from Ref.~\onlinecite{flenner:2011}). Whereas the change in the static 
length is measured in fractions of their low-density value, $\xi_\mathrm{dyn}$ grows by a factor of  $4.5$ when reaching 
$\beta p/\rho= 22$, and seems to keep on growing~\cite{flenner:2011}.

We can verify that the static length $\xi_\mathrm{p}$ extracted satisfies the bound derived by Montanari and Semerjian~\cite{montanari:2006} that relates the structural relaxation time scale and the (cavity) point-to-set length $\xi_{\mathrm{PS}}$. One may indeed reasonably expect that $\xi_{\mathrm{p}}\sim\xi_{\rm PS}$, as long as one is far from any putative random first-order transition~\cite{cammarota:2011a}, which is the case here. (When approaching a random first-order transition, one should have that $\xi_{\mathrm{p}}\sim \xi_{\rm PS}^{1/d}$.) The bound states that
\begin{equation}
\tau_\alpha \lesssim \tau_{0}\exp\left (B\, \xi_{\rm PS}^d \right),
\end{equation}
where $\tau_{0}$ is a constant setting the microscopic time scale. The coefficient $B$ depends on pressure (and temperature, more generally) and is such that when $\xi_{\rm PS}\sim\sigma$ the right-hand side describes the ``noncooperative dynamics'' of the model~\cite{franz:2011}. Using an Arrhenius-like argument for activation volumes~\cite{berthier:2009}, we note that, all else being equal, higher pressures trivially rescale the free-energy landscape and thereby slow the dynamics. For a hard-sphere fluid, one then expects $B \propto \beta P$.  The upper bound of $\tau_\alpha$ diverges with the pressure even in the absence of any growing $\xi_{\rm PS}$, as when approaching $T=0$ for an Arrhenius temperature dependence. In the low and moderate density fluid, the relaxation time indeed follows $\tau_\alpha(P)\simeq \tau_{\rm low}(P)=\tau_0\exp[K \beta P]$ with $K$ a density-independent constant.  One then finds that
\begin{equation}
\label{eq_bound-length}
\frac{\xi_{\rm PS}(P)}{\xi_{\rm PS,0}} \gtrsim \left(\frac{\log [\tau_\alpha(P)/\tau_{0}]}{\log[\tau_{\rm low}(P)/\tau_0]}\right)^{1/d},
\end{equation}
where $\xi_{\rm PS,0}$ is the low-density limit of $\xi_{\rm PS}$. Equation~(\ref{eq_bound-length}) thus provides a lower bound for the growth of a static length imposed by the dynamical slowdown. As seen in Fig.~\ref{fig:dynlength}, the bound given by Eq.~(\ref{eq_bound-length}) increases only slowly in the dynamically accessible domain and is already satisfied by $\xi_{\rm p}$.

In dimension higher than 3, we are unable to directly compute the point-to-set length $\xi_{\rm p}$, but we can nonetheless use the bound in Eq.~(\ref{eq_bound-length}) as a proxy for the growth of a static length. The results in Fig.~\ref{fig:dynlength} give strong evidence that the growth of any static length is very limited for hard-sphere glass-formers in the accessible dynamical range where the relaxation time increases by 3 to 4 orders of magnitude.

\subsection{Linear response for the overlap at small concentration of pinned particles}

We conclude this section on order-agnostic static lengths by proving that one cannot gain any knowledge about point-to-set correlations from a mere investigation of the small-concentration pinning regime. In this limit, $Q_c(\infty)-Q_0(\infty)$ indeed only conveys information about the equilibrium fluid pair correlation function. Higher-order terms in the expansion in powers of $c$ involve higher-order equilibrium correlation functions of the fluid, the order $p$ of the expansion being typically expressible in terms of $n$-body correlation functions with $n\leq p$.

The overlap between the original equilibrium configuration and that equilibrated in the presence of pinned particles after due average can be expressed in the long-time limit as 
\begin{equation}
\label{eq_overlap_infty}
\begin{aligned}
&Q_c(\infty)=
\\&\frac{1}{(1-c)^2N} \int d\mathbf{ r}\int d\mathbf{ r'}\, \Theta(a-|\mathbf{r}-\mathbf{r'}|) \left\langle \overline{\hat\rho(\mathbf r)\left\langle \hat\rho(\mathbf{r'})\right\rangle_{\mathcal B}} \right\rangle ,
\end{aligned}
\end{equation}
where $\hat\rho(\mathbf r)=\sum_{n \notin \mathcal B} \delta(\mathbf r - \mathbf{r}_n)$ is the local density of  unpinned particles 
and $\left\langle \; \right\rangle_{\mathcal B}$ represents a conditional average over configurations of particles at equilibrium with a 
set $\mathcal B$ of pinned particles. Specifically,
\begin{equation}
\label{eq_condit_rho}
\begin{aligned}
\left\langle \hat\rho(\mathbf{r'})\right\rangle_{\mathcal B}=
\int d \mathbf{\tilde r}^N\, \frac{ e^{-\beta V(\{\mathbf{\tilde r}^N\})}}{Z_N(\{\mathbf r_m\},\mathcal B)} \prod_{k\in \mathcal B} \delta(\mathbf r_k - \mathbf{\tilde r}_k)\, \hat\rho(\mathbf{r'})  ,
\end{aligned}
\end{equation}
where
\begin{equation}
\label{eq_condit_Z}
\begin{aligned}
Z_N(\{\mathbf r_m\},\mathcal B)=
\int d \mathbf{\tilde r}^N\,  e^{-\beta V(\{\mathbf{\tilde r}^N\})} \prod_{k\in \mathcal B} \delta(\mathbf r_k - \mathbf{\tilde r}_k) ,
\end{aligned}
\end{equation}
$V(\{\mathbf{\tilde r}^N\})$ denote the interaction potential, which corresponds to pairwise hard-core terms here, and $\hat\rho(\mathbf r')=\sum_{m \notin \mathcal B} \delta(\mathbf{r'} - \mathbf{\tilde r}_m)$. One should keep in mind that  $\left\langle \hat\rho(\mathbf{r'})\right\rangle_{\mathcal B}$ still depends on the reference configuration $\{\mathbf r_m\}$ and on the pinned set $\mathcal B$, the integration being on the configurations denoted $\{\mathbf{\tilde r}^N\}\equiv \{\mathbf{\tilde r}_m\}$.

The conditional probability that is used in Eq. (\ref{eq_condit_rho}) satisfies the requirement that when integrated with respect to the equilibrium reference configuration it gives back the unconstrained equilibrium distribution~\cite{franz:2011,scheidler:2004,krakoviack:2010}. As an illustration, one obtains
\begin{equation}
\label{eq_average_condit_rho}
\begin{aligned}
&\left\langle\left\langle \hat\rho(\mathbf{r'})\right\rangle_{\mathcal B}\right\rangle=\int d \mathbf{ r}^N\, \frac{ e^{-\beta V(\{\mathbf{ r}^N\})}}{Z_N}
\int d \mathbf{\tilde r}^N\, \frac{ e^{-\beta V(\{\mathbf{\tilde r}^N\})}}{Z_N(\{\mathbf r_m\},\mathcal B)} \times 
\\& \prod_{k\in \mathcal B} \delta(\mathbf r_k - \mathbf{\tilde r}_k)\, \hat\rho(\mathbf{r'}) = 
\int d \mathbf{\tilde r}^N\, \frac{ e^{-\beta V(\{\mathbf{\tilde r}^N\})}}{Z_N}\, \hat\rho(\mathbf{r'})\times \\&
\int d \mathbf{ r}^N\, \frac{ e^{-\beta V(\{\mathbf{ r}^N\})}}{Z_N(\{\mathbf{\tilde{r}}_m\},\mathcal B)} \prod_{k\in \mathcal B} \delta(\mathbf r_k - \mathbf{\tilde r}_k) =\\& \int d \mathbf{\tilde r}^N\, \frac{ e^{-\beta V(\{\mathbf{\tilde r}^N\})}}{Z_N}\, \hat\rho(\mathbf{r'}) = 
\left\langle \hat\rho(\mathbf{r'})\right\rangle,
\end{aligned}
\end{equation}
where use has been made of the delta functions to change the arguments of $Z_N(\{\mathbf r_m\},\mathcal B)$ and of the definition of the constrained partition function in Eq. (\ref{eq_condit_Z}).

To organize the expansion in powers of the pinning fraction $c$, rather than considering realizations with exactly $cN$ pinned particles, it is easier to consider the situation in which $cN$ particles are pinned \textit{on average}, as in a grand-canonical treatment~\cite{krakoviack:2010,jack:2012}. Let $\boldsymbol{\tau}$ be the set of $N$ occupation variables $\{ \tau_i=0,1\}$  characterizing the specific pinning of particles in a given configuration. The corresponding probability distribution is then simply
\begin{equation}
\label{eq_random_distrib}
P_N(\boldsymbol\tau)=(1-c)^{N} \left (\frac{c}{1-c}\right )^{\sum_{i=1}^N \tau_i}.
\end{equation}
We then rewrite the two-point density appearing in the expression of the overlap $Q_c(\infty)$ as
\begin{equation}
\label{eq_overlap_infty_GC}
\begin{aligned}
&\left\langle \overline{\hat\rho(\mathbf r)\left\langle \hat\rho(\mathbf{r'})\right\rangle_{\mathcal B}} \right\rangle= \mathrm{tr}_{\boldsymbol\tau} \bigg\{ P_N(\boldsymbol\tau) \int d \mathbf{ r}^N\, \frac{ e^{-\beta V(\{\mathbf{ r}^N\})}}{Z_N} \hat\rho(\mathbf{r;\boldsymbol\tau}) \times \\&
\int d \mathbf{\tilde r}^N\, \frac{ e^{-\beta V(\{\mathbf{\tilde r}^N\})}}{Z_N(\{\mathbf r_m\},\mathcal B(\boldsymbol\tau))} \prod_{k\in \mathcal B(\boldsymbol\tau)} \delta(\mathbf r_k - \mathbf{\tilde r}_k)\, \hat\rho(\mathbf{r';\boldsymbol\tau})\bigg\} ,
\end{aligned}
\end{equation}
where we have now explicitly indicated the dependence on the pinned set $\boldsymbol\tau$.

The desired expansion in  $c$ can be generated by first formally expanding the average over the pinned set in increasing 
number of pinned particles. For a generic function $\mathcal F(\boldsymbol\tau)$, this reads
\begin{equation}
\label{eq_random_distrib_expand}
\begin{aligned}
\overline{\mathcal F(\boldsymbol\tau)}=&\mathrm{tr}_{\boldsymbol\tau}  \bigg \{  \sum_{p \geq 0} \frac{1}{p!} c^{p}(1-c)^{N-p} 
\sum_{k_1,..,k_p=1}^N \prod_{i=1}^{N} \delta(\tau_{k_i},1) \\& \times \prod_{j \notin \{k_1,..,k_p \}} \delta(\tau_j,0)\; \mathcal F(\boldsymbol\tau) \bigg \}, 
\end{aligned}
\end{equation}
where $\delta(a,b)$ is the Kronecker symbol. After expanding the remaining factors in $(1-c)$, one finally arrives at
\begin{equation}
\label{eq_random_distrib_expandfinal}
\begin{aligned}
\overline{\mathcal F(\boldsymbol\tau)}=&\mathrm{tr}_{\boldsymbol\tau}  \bigg \{ \bigg [\prod_{j=1}^{N} \delta(\tau_{j},0) + c \sum_{k=1}^N (\delta(\tau_{k},1)-\delta(\tau_{k},0))\times \\&\prod_{j\neq k} \delta(\tau_{j},0)+ \frac{c^2}{2} \sum_{k_1,k_2=1}^N (\delta(\tau_{k_1},1)-\delta(\tau_{k_1},0))\times \\&(\delta(\tau_{k_2},1)-\delta(\tau_{k_2},0))\prod_{j \neq k_1,k_2} \delta(\tau_j,0) + \cdots \bigg ]  \mathcal F(\boldsymbol\tau) \bigg \},
\end{aligned}
\end{equation}
where only the first terms have been explicitly given.

Eq. (\ref{eq_random_distrib_expandfinal}) should now be applied to $\left\langle  \hat\rho(\mathbf r)\left\langle \hat\rho(\mathbf{r'})\right\rangle_{\mathcal B}\right\rangle$.  The zeroth-order term corresponds to the case with no pinning ($\mathcal B$ is empty and $\hat\rho(\mathbf r)\equiv\rho(\mathbf r)$) and simply gives $\left\langle  \rho(\mathbf r)\left\langle \rho(\mathbf{r'})\right\rangle\right\rangle=\rho^2$. The linear term can be expressed as
\begin{equation}
\label{eq_linear_term1}
\begin{aligned}
&\sum_{k=1}^N \sum_{m,n\neq k} \int d \mathbf{r}^N \, \frac{ e^{-\beta V(\{\mathbf{ r}^N\})}}{Z_N} \delta(\mathbf{r}-\mathbf{r}_m) \int d \mathbf{\tilde r}^N \frac{ e^{-\beta V(\{\mathbf{\tilde r}^N\})}}{Z_N(\{\mathbf r_m\},k)}\\&  \delta(\mathbf r_k - \mathbf{\tilde r}_k)\delta(\mathbf{r'}-\mathbf{\tilde r}_n) - 
N\, \sum_{m,n} \int d \mathbf{r}^N \, \frac{ e^{-\beta V(\{\mathbf{ r}^N\})}}{Z_N} \delta(\mathbf{r}-\mathbf{r}_m)\\& \int d \mathbf{\tilde r}^N \frac{ e^{-\beta V(\{\mathbf{\tilde r}^N\})}}{Z_N}\delta(\mathbf{r'}-\mathbf{\tilde r}_n),
\end{aligned}
\end{equation}
where because of the fluid's translational invariance 
\begin{equation}
\label{eq_condit_Zlinear}
\begin{aligned}
Z_N(\{\mathbf r_m\},k)=
\int d \mathbf{\tilde r}^N\,  e^{-\beta V(\{\mathbf{\tilde r}^N\})} \delta(\mathbf r_k - \mathbf{\tilde r}_k)=\frac{Z_N}{V}.
\end{aligned}
\end{equation}

By using the definition of the $1$- and $2$-body densities of the fluid in the absence of pinning~\cite{hansen:1986},
\begin{equation}
\label{eq_1-fluiddensity}
\begin{aligned}
&N \int \prod_{i\neq k}d \mathbf{r}_i \, \frac{ e^{-\beta V(\{\mathbf{ r}^N\})}}{Z_N} =\rho^{(1)}(\mathbf{r}_k)=\rho\, ,\\&
N(N-1) \int \prod_{i\neq k,k'}d \mathbf{r}_i \, \frac{ e^{-\beta V(\{\mathbf{ r}^N\})}}{Z_N}  =\rho^{(2)}(\vert \mathbf{r}_k - \mathbf{r}_{k'} \vert)\, ,
\end{aligned}
\end{equation}
one can rewrite Eq. (\ref{eq_linear_term1}) as
\begin{equation}
\label{eq_linear_term2}
\begin{aligned}
&\frac{V}{N(N-1)^2} \sum_{m,n\neq 1} \int d \mathbf{r}_1 \int d \mathbf{r}_m \int d \mathbf{\tilde r}_n \, \delta(\mathbf{r}-\mathbf{r}_m) \delta(\mathbf{r'}-\mathbf{\tilde r}_n)\\& \times \rho^{(2)}(\vert \mathbf{r}_1 - \mathbf{r}_{m} \vert)\rho^{(2)}(\vert \mathbf{r}_1 - \mathbf{\tilde r}_{n} \vert)- N \rho^2\, .
\end{aligned}
\end{equation}
Because a cancellation of the leading terms in $N$ is anticipated, one has to be cautious about the sub-extensive terms in the long-distance limit of the pair density, \textit{i.e.} $\rho^{(2)}( r)\rightarrow \rho^2(1-a/N)$ when $r\rightarrow \infty$. (For an ideal gas, $\rho^{(2)}( r)\rightarrow N(N-1)/V^2$ so that $a=1$.) In consequence, we define $\rho^{(2)}( r)= \rho \delta(\mathbf r) + \rho^2[1-a/N + h( r)]$ with $h( r \rightarrow \infty)=0+ \mathcal O(1/N^2)$ -- compare with Eq. (\ref{eq_pairdensity}). Noting that the $\delta$ terms in the pair densities do not contribute, Eq. (\ref{eq_linear_term2}) can then be expressed as 
\begin{equation}
\label{eq_linear_term4}
\begin{aligned}
- 2 \rho^2 \left [a-\rho \int  d \mathbf{r}  h(r)\right ] + \rho^3 \int  d \mathbf{r''} h(\vert \mathbf{r''} - \mathbf{r} \vert)h(\vert \mathbf{r''} - \mathbf{r'} \vert)
\end{aligned}
\end{equation}
where sub-extensive corrections  of $\mathcal O(1/N)$ have been dropped. We now determine the coefficient $a$ by requiring that $\left\langle  \hat\rho(\mathbf r)\left\langle \hat\rho(\mathbf{r'})\right\rangle_{\mathcal B}\right\rangle \rightarrow (1-c)^2\, \rho^2$ when $\vert \mathbf r - \mathbf{r'}\vert \rightarrow \infty$. When expanded in powers of $c$, this condition implies that $a-\rho \int  d \mathbf{r}  h( r)=1$ (as  indeed found for the ideal gas).

After inserting the above results for the zeroth and first orders in Eq. (\ref{eq_overlap_infty}), subtracting the value of $Q_0(\infty)$, and expanding the term in $(1-c)^2$ in the denominator of the right-hand side of Eq. (\ref{eq_overlap_infty}), one finally obtains
\begin{equation}
\label{eq_linear_final}
\begin{aligned}
&Q_c(\infty)- Q_0(\infty)=\\& c \,\rho^2 \int \int d \mathbf{r} d \mathbf{r'} \, \Theta(a-r)h( r')h(\vert \mathbf{r'}-\mathbf{r}\vert) +\mathcal O(c^2)\, .
\end{aligned}
\end{equation}
For an ideal gas ($h\equiv 0$), one recovers the expected result, $Q_c(\infty)- Q_0(\infty)=0$.

We have focused here on the linear term in $c$ which expresses the linear response of the fluid to a perturbation associated with pinning a vanishingly small concentration of particles. The higher-order terms in powers of $c$ can be similarly computed, but lead to increasingly tedious algebra. It is easily realized that the $p$th order then involves up to $p$-point density correlation functions.

As is clear from the above expression, the linear coefficient of the dependence of the configuration overlap on the concentration of pinned particles does not carry any relevant information on point-to-set correlations. It only involves the pair density correlations. As a result, one can at best hope to extract from it the corresponding two-point density correlation length, a quantity that is not of much interest in the context of glass-forming fluids.

In Appendix \ref{sec:blocking}, we also relate the configuration overlap to the ``blocking'' (or ``disconnected'') correlation function introduced in the statistical mechanics of fluids in a disordered porous medium, and more generally in the theory of disordered systems.

\section{Conclusion}
In summary, our results demonstrate that traditional descriptions of geometrical frustration are of limited applicability to hard-sphere and related systems, and that, more generally, structural correlations are fairly limited in the dynamical regime accessible in simulations. The complex and related question of the relationship between the structural and the dynamical correlation lengths will be the subject of a future publication.


\begin{acknowledgments}
We acknowledge stimulating interactions with L. Berthier, G. Biroli, E. Corwin, C. Cammarota, D. Nelson, R. Mosseri, Z. Nussinov, D. Reichman, and R. Schneider. PC acknowledges NSF support No.~NSF DMR-1055586. BC received NSERC funding.  This work was made possible by the facilities of the Shared Hierarchical 
Academic Research Computing Network (SHARCNET:www.sharcnet.ca) and Compute/Calcul Canada.
\end{acknowledgments}

\appendix

\section{Simplex Coordination and Schl\"afli Notation}
\label{sect:appschaefli}
Given a polytope or honeycomb $Z$, let's denote its $j$-skeleton by $Z^{(j)}$.  Hence $Z^{(j)}$ is the set of all the $j$-dimensional faces of $Z$. For $F\in Z^{(j)}$,  we let 
	\[N_{j,k}(F,Z)=\#\{S\in Z^{(k)}\mid S\cap F\neq \emptyset\},\]
and let $n_{j,k}(Z)$ be the average value of $N_{j,k}(F,Z)$ amongst all $F\in Z^{(j)}$.
For a regular polytope or honeycomb, $N_{j,k}(F,Z)$ is constant and equals $n_{j,k}(Z)$.

\begin{lemma}Let $Z$ be a regular polytope or honeycomb in $\R^d$ or  $S^d$ given by the Schl\"afli symbol $\{a_1,\ldots,a_d\}$.  Then $\bar{q}_{d-2,d}=a_d$. 
\end{lemma}

\begin{proof}
Let $p\in Z^{(0)}$ be a vertex of $Z$ and let $Z_p$ be all the faces of $Z$ containing $p$.  If $F\in Z^{(j)}_p$ and $S\in Z^{(i)}$ are such that $F\subseteq S$ then $S\in Z^{(i)}_p$.  So when $i\geq j$ and $F\in Z_p$,
we have $n_{j,i}(F,Z)=n_{j,i}(F,Z_p)$.

Let $V_p$ be the vertex figure of $p$. Let's use the convention that $V^{(-1)}_p$ is composed of the empty face only. There is an inclusion preserving bijection between the skeleton $\bigcup_{i\geq 0}Z^{(i)}_p$ of $Z_p$ and the skeleton $\bigcup_{i\geq -1}V^{(i)}_p$ of $V_p$, lowering the dimension of the faces by $1$.   We have
\begin{align*}
	\bar{q}_{d-2,d}(Z)&=n_{d-2,d}(Z)
	=n_{d-2,d}(Z_p)
	=n_{d-3,d-1}(V_p),
\end{align*}
as long as $d\geq 2$.
So 
\begin{align*}
	\bar{q}_{d-2,d}(\{a_1,\ldots,a_d\})&=n_{d-3,d-1}(\{a_{2},\ldots,a_{d}\})\\
	&=\cdots
	=n_{-1,1}(\{a_d\})\\
	&=\#\text{ edges in }\{a_d\}
	=a_d.
\end{align*}
\end{proof}

This method can even be used  for non-regular honeycombs.  Consider the lattice $E_8$ in $\R^8$.    Though the lattice is not regular, each vertex figure is $4_{21}$ (see Section 21.3D of Ref.~\onlinecite{conway:1988}).  One can thus iteratively obtain that $\bar{q}_{6,8}(E_8)=n_{5,7}(4_{21})=\cdots=n_{1,3}(0_{21})=n_{0,2}(\text{triangular prism})=3$.

Being able to use one of the number in the Schl\"afli symbol to detect the coordination number $\bar{q}_{d-2,d}$ is a fortunate event.
One could not  obtain  $\bar{q}_{d-3,d}$ so simply, for instance.  The same reasoning indeed yields
$\bar{q}_{0,3}(\{3,3,5\})=n_{-1,2}(\{3,5\})=20$ and $\bar{q}_{0,3}(\{5,3,3\})=4$, results which are less trivial to read off the Schl\"afli symbol.


\section{relationship between $\bar{Z}$ and coordination numbers}\label{sect:appeuler}
Let $f_j$ be the number of $\s_j$ in a $\s_d$, i.e., the number of faces of dimension $j$.
We use $q_j\equiv q_{j,d}$ and $Z\equiv Z_d$ for notational simplicity. 

\begin{lemma}\label{thm:coupdepied} We have
$f_j N_d=\bar{q}_j N_j$, for $0\leq j\leq d$ and $N_1=N\bar{Z}/2$.
\end{lemma}

\begin{proof}
The second relation is an immediate consequence of the classical hand shake lemma: in a graph, twice the number of edges must equal the sum of the degrees of the vertex, so $2N_1= \sum_i Z(i)=N\bar{Z}$. Similarly, if each $\s_d$ in the decomposition has a number $f_j$ of faces $\s_j$, we must have
$f_j N_d=\sum_i q_j(i)=\bar{q}_j N_j$. 
\end{proof}

\begin{theorem}\label{thm:Zbar}Let $\chi$ be the Euler characteristic of the $d$-dimensional manifold in which we compute the coordination numbers.  Then 
	\begin{equation*}
	\bar{Z}_d=\frac{d(d+1)(1-\chi/N)}{\binom{d+1}{2}-\bar{q}_{1,d}\left((-1)^{d-1}\frac{d-1}2+\sum_{j=2}^{d-2} \frac{(-1)^j\binom{d+1}{j+1}}{\bar{q}_{j,d}}\right)}.	
	\end{equation*}
\end{theorem}

\textbf{Note:} The Euler characteristic is automatically zero when $d$ is odd or when the manifold is a flat torus.  As discussed earlier, the Euler characteristic is always zero in this paper.

\begin{proof}
 We have
\begin{align*}
\chi&=\sum_{j=0}^d (-1)^j N_j\displaybreak[0]\\
&=N_0-N_1+\sum_{j=2}^d (-1)^j \frac{f_j N_d}{\bar{q}_j}\displaybreak[0]\\
&= N-\frac{N\bar{Z}}{2}+\sum_{j=2}^d (-1)^j \frac{f_j N_d}{\bar{q}_j}\displaybreak[0]\\
&= N-\frac{N\bar{Z}}{2}+\sum_{j=2}^{d} (-1)^j \frac{f_j}{\bar{q}_j}\frac{\bar{q}_1}{f_1} N_1\displaybreak[0]\\
&= N-\frac{N\bar{Z}}{2}\left(1-\sum_{j=2}^d (-1)^j \frac{f_j}{f_1}\frac{\bar{q}_1}{\bar{q}_j}\right).
\end{align*}

Note that $f_j=\binom{d+1}{j+1}$, in particular $f_0=d+1$ and $f_d=1$. Note also that, obviously $\bar{q}_d=1$, and $\bar{q}_{d-1}=2$. Solving, we get
\[\bar{Z}=\frac{2(1-\chi/N)}{\displaystyle1-\Bigl(\sum_{j=2}^d (-1)^j\frac{\binom{d+1}{j+1}}{\binom{d+1}{2}} \bar{q}_j^{-1}\Bigr)\bar{q}_1} \]
which simplifies to the desired formula.
\end{proof}

Note that $\bar{q}_0$ doesn't appear in this formula.  An alternative formula is
\begin{equation}\label{eqn:Zbar3}
\bar Z=d\frac{\bar{q}_0}{\bar{q}_1}.
\end{equation}

Indeed, repeated use of lemma \ref{thm:coupdepied} tells us that 
	\begin{align*}
		\bar{q}_0&=\frac{f_0}{N_0}N_d
		 =\frac{f_0}{N}\frac{\bar{q}_1N_1}{f_1}
		=\frac{f_0}{f_1}\frac{\bar{q}_1\bar Z}{2}
		=\frac{\binom{d+1}{1}}{2\binom{d+1}{2}}\bar{q}_1\bar Z
		=\frac{1}{d}\bar{q}_1\bar Z.
	\end{align*}

\section{ideal coordination numbers} \label{app:ideal}
In this appendix, we cover two aspects of the computation.  First, we prove Theorem \ref{prop:iteration}.  
Second, we prove the claims left unproved in the establishment of Eqn (\ref{eqn:qideal}).

\subsection{Proof of Theorem \ref{prop:iteration}}
It is useful for the proof to introduce one more variable: the radius of the sphere.
Given $n$ vectors $\vec{v}_1,\ldots, \vec{v}_n\in\R^n$ such that 
\[\vec{v}_i\cdot \vec{v}_j=\begin{cases}k,&\text{ if } i\neq j,\\ 1,& \text{ if }i=j,\end{cases}\]
we define the regions
\begin{align*}
 T(k,n,p)&\equiv \{\vec{x}\in\R^n\mid \vec{x}\cdot\vec{v}_i\geq p \text{ for all }i\},\\	
S(R,k,n,p)&\equiv \{\vec{x}\in\R^n\mid \|\vec{x}\|=R, \text{ and } \vec{x}\in T(k,n,p) \},
\end{align*}
and let
\[V(R,k,n,p)\equiv\Vol(S(R,k,n,p)).\]
To conform with the notation in the main text, we have $\tilde V(k,n,p)=V(1,k,n,p)$.

\begin{lemma}\label{lemma:removeR}We have $V(R,k,n,p)=R^{n-1}V(1,k,n,\frac pR)$.
\end{lemma}

\begin{proof}  If $\vec{x}\in S^{n-1}(R)$ satisfies $\vec{x}\cdot \vec{v}_i\geq p$ then $\frac {\vec{x}}R \in S^{n-1}(1)$ satisfies $\frac {\vec{x}}R\cdot \vec{v}_i\geq \frac pR$, and vice-versa.  So the map $\vec{x}\mapsto \frac {\vec{x}}R$ gives a bijection between  $S(R,k,n,p)$ and $S(1,k,n,\frac pR)$.  This map shrinks hypersurface volumes by a factor of $R^{n-1}$.
\end{proof}

We can arbitrarily rotate the $\vec{v}_i$, so let us assume that $\vec{v}_n=\vec{e}_n=(0,\ldots,0,1)$.

\begin{lemma}Suppose $\vec{v}_n=\vec{e}_n$. For $(\vec{y},z)\in \R^{n-1}\times \R$, we have
	\[(\vec{y},z)\in T(k,n,p) \iff \left.\begin{matrix}z\geq p,\text{ and }\\ \vec{y}\in T\bigl(\frac{k}{1+k},n-1,\frac{p-kz}{\sqrt{1-k^2}}\bigr).\end{matrix}\right.\]
\end{lemma}

\begin{proof}
Since $\vec{v}_n=\vec{e}_n$, we have that $\vec{v}_i\cdot \vec{e}_n=k$ for all $i\leq (n-1)$.  
Let $\vec{w}_i\equiv \frac{\vec{v}_i-k\vec{e}_n}{\|\vec{v}_i-k\vec{e}_n\|}$, so that $\vec{w}_i\in\R^{n-1}$. 
Since $\|\vec{v}_i-k\vec{e}_n\|^2=1-2k^2+k^2=1-k^2$, $\vec{v}_i=\sqrt{1-k^2}\vec{w}_i+k\vec{e}_n$.  Moreover, for $i\neq j$, we have $k=\vec{v}_i\cdot\vec{v}_j=(1-k^2)\scp{\vec{w}_i,\vec{w}_j}+k^2$, hence
\[\vec{w}_i\cdot\vec{w}_j=\frac{k-k^2}{1-k^2}=\frac{k}{1+k}.\]

Note that for $i<n$,  $\vec{x}\cdot\vec{v}_i=(\vec{y}+z\vec{e}_n)\cdot (\sqrt{1-k^2}\vec{w}_i+k\vec{e}_n)=\sqrt{1-k^2}\vec{y}\cdot\vec{w}_i+kz$.  Hence the conditions  $\vec{x}\cdot\vec{v}_i\geq p$ are equivalent to 
\[\vec{y}\cdot\vec{w}_i\geq  \frac{p-kz}{\sqrt{1-k^2}}, \text{ for }i<n.\]
The condition $\vec{x}\cdot\vec{v}_n\geq p$ is equivalent to $z\geq p$. 
\end{proof}

\begin{lemma}\label{lemma:changeframe-S}Suppose $\vec{v}_n=\vec{e}_n$. Then $(\vec{y},z)\in S(k,n,p)$ if and only if
\begin{gather*}
	p\leq z\leq h(k,n,p),\text{ and }\\
	\vec{y}\in S\bigl(\sqrt{1-z^2},\frac{k}{1+k},n-1,\frac{p-kz}{\sqrt{1-k^2}}\bigr).
\end{gather*}
\end{lemma}

\begin{proof}
 Let's start with the easy part.  Let $\vec{x}=(\vec{y},z)\in\R^{n-1}\times \R$.  If $|\vec{x}|=1$, then $|\vec{y}|=\sqrt{1-z^2}$. It is also clear that we must have $z\geq p$, at least from the previous lemma.

Geometrically, we see that for $(\vec{y},z)$ to be in that region, the condition $z\geq p$ is insufficient.  The ``tip'' of the region occurs at a point $\vec{X}$ where 
$\vec{X}\cdot\vec{v}_i=p$
for  $1\leq i<n$, and  $\vec{X}\cdot\vec{e}_n=h$, and $|\vec{X}|=1$.

Let $g$ be the matrix with entries $g_{ij}=\vec{v}_i\cdot\vec{v}_j$.  Let $\hat{\vec{X}}$ be the column vector of the coordinates of $\vec{X}$ in the basis $\vec{v}_i$.  So $\vec{X}=\sum \hat{\vec{X}}_i \vec{v}_i$.  Let $\vec{Y}=\vect{p&\cdots &p&h}$.  Then 
\[\vec{Y}_i=\vec{X}\cdot\vec{v}_i
	=\sum_j \hat{\vec{X}}_j \vec{v}_j\cdot \vec{v}_i
	=\sum_j g_{ij}\hat{\vec{X}}_j,\]
so $\vec{Y}=g\hat{\vec{X}}$, or $\hat{\vec{X}}=g^{-1}\vec{Y}$.  Then 
\[	1=\vec{X}\cdot\vec{X}
	=\sum_{i,j} \hat{\vec{X}}_i\hat{\vec{X}}_jg_{ij}
	=\hat{\vec{X}}^T g\hat{\vec{X}}
	=\vec{Y}^T g^{-1} \vec{Y}.
\]
So $h$ is a solution to the quadratic equation $\vec{Y}^Tg^{-1}\vec{Y}=1$.  There are two solutions to this equation. We need to take the solution greater than $p$.  The function $h=h(k,n,p)$ is that solution.\end{proof}

We have now set up the stage to perform a volume computation in cylindrical coordinates.
One has to be quite careful in using cylindrical coordinates on the sphere. 
\begin{lemma}We have the relation
\begin{equation*}\label{eqn:cylindrical}
	\dvol_{S^n(R)}=\frac{R}{\sqrt{R^2-z^2}}\dvol_{S^{n-1}(\sqrt{R^2-z^2})}dz
\end{equation*}
between the volume element on the sphere of radius R centeredcenter at the origin of $\R^{n+1}$ and the volume element on the parallel sphere at height $z$. 
\end{lemma}

\begin{proof}
Let $s_n\colon \R^{n-1}\to \R^{n}$ be the inverse of the stereographic projection, and
\begin{align*}
	q\colon \R^{n-1}\times [-R,R]&\to \R^{n+1}\\
	(x\quad ,\quad t)\quad \quad&\mapsto (\sqrt{R^2-t^2}s_n(x),t).
\end{align*}
The image of $q$ is all of $S^n(R)$ less half of a great circle.

Recall that $(s_n^*g_{S^{n-1}})_x=(ds_n)_x^T(ds_n)_x$.
Similarly, we have that $(q^*g_{S^n(R)})=(dq)^T(dq)$.  But
\[(dq)_{(x,t)}=\vect{\sqrt{R^2-t^2}(ds_n)_x & \frac{-t}{\sqrt{R^2-t^2}}s_n(x)\\
0&1},\]
so the metric is 
\begin{align*}
	g=(dq)^T(dq)
	&=\vect{(R^2-t^2)s_n^*g_{S^{n-1}(1)} & -t (ds_n)^T_x s_n(x)\\ -t s_n(x)^T(ds_n)_x&\frac{t^2}{R^2-t^2}+1}\\
	&=\vect{(R^2-t^2)s_n^*g_{S^{n-1}(1)} & 0\\ 0 & \frac{R^2}{R^2-t^2}}.
\end{align*}
By taking the determinant and a square root, we get
\[\dvol_{S^n(R)}=R(R^2-t^2)^{\frac{n-2}2} \dvol_{S^{n-1}(1)}dt,\]
as desired.  Rescaling, we establish Eqn~(\ref{eqn:cylindrical}). \end{proof}

\begin{widetext}
\begin{proposition}\label{prop:iteration1}
For $n>2$, we have
\begin{equation*}
	\tilde V(k,n,p)=\int_p^{h(k,n,p)}\!\! (1-z^2)^{\frac{n-3}2} \tilde V\bigl(\frac{k}{1+k},n-1,\frac{p-kz}{\sqrt{1-k^2}\sqrt{1-z^2}}\bigr)dz.
\end{equation*}
\end{proposition}
\end{widetext}

\begin{proof} Let $f_1,f_2,f_3(\cdot, z)$ be the characteristic functions of the set $S(k,n,p)$, the interval $[p,h(k,n,p)]$, and the set $S\bigl(\sqrt{1-z^2},\frac{k}{1+k},n-1,\frac{p-kz}{\sqrt{1-k^2}}\bigr)$ respectively.  So $f_1(\vec{y},z)=f_2(z)f_3(\vec{y},z)$ by Lemma \ref{lemma:changeframe-S}.  We have
\begin{align*}
	\tilde V(k,n,p)&=\int_{S^{n-1}}\!\!\! \!\!\! f_1(\vec{y},z)\dvol_{S^{n-1}}(\vec{y},z)\\
	&=\int_{-1}^1\int_{S^{n-2}}\!\!\! \frac{f_2(z)f_3(\vec{y},z)}{\sqrt{1-z^2}}\dvol_{S^{n-2}(\sqrt{1-z^2})}(\vec{y})  dz\\
	&=\int_p^{h(k,n,p)}\! \frac{V\bigl(\sqrt{1-z^2},\frac{k}{1+k},n-1,\frac{p-kz}{\sqrt{1-k^2}}\bigr)}{\sqrt{1-z^2}}dz.
\end{align*}
Then, using Lemma \ref{lemma:removeR}, we obtain the desired result.
\end{proof}

To conclude the proof of Theorem \ref{prop:iteration}, we need the following two lemmas.
\begin{lemma}\label{lemma:arccos}
$\tilde V(k,2,0)=\arccos(-k)$.
\end{lemma}

\begin{lemma}\label{lemma:f}$
		\tilde V(k,2,p)=\tilde V(f(k,p),2,0)$.
\end{lemma}

\begin{proof}Figure \ref{fig:Dim2} illustrates the idea of this proof.
	\begin{figure}
	\center{
	\includegraphics[width=0.8\columnwidth]{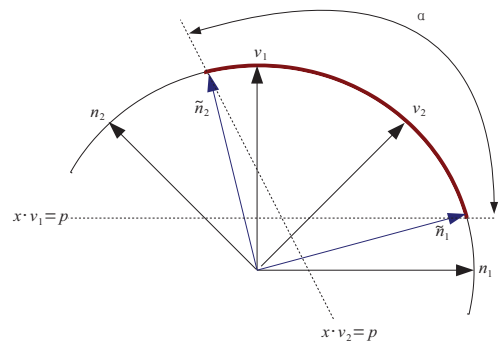}}
	\caption{Illustration for the proof of Lemma \ref{lemma:f}}
	\label{fig:Dim2}
	\end{figure}
Let  $\vec{n}_1$, $\vec{n}_2$ be unit vectors perpendicular  to $\vec{v}_1$ and $\vec{v}_2$ respectively, pointing in the direction of the region we are computing.  Let $\hat{\vec{n}}_1$ to be the vector closest to $\vec{n}_1$ on the circle $S^1$ and in the plane $\vec{x}\cdot \vec{v}_1=p$. Let $\hat{\vec{n}}_2$ be similarly defined.  Let $\hat{\vec{v}}_1$ and $\hat{\vec{v}}_2$ be the vectors such that 
\[\vec{x}\cdot\vec{v}_i\geq p\iff \vec{x}\cdot\hat{\vec{v}}_i\geq 0.\]
We have $\hat{\vec{v}}_i\cdot\hat{\vec{n}}_i=0$.
Let $k^*=\sqrt{1-k^2}$.

Let $\alpha=\mathrm{angle}(\vec{v}_1,\vec{v}_2)$ and $\theta=\mathrm{angle}(\vec{n}_1,\vec{n}_2)$. We have $\alpha=\pi-\theta$.  Thus \[\vec{n}_1\cdot \vec{n}_2=\cos(\theta)=\cos(\pi-\alpha)=-\cos(\theta)=-\vec{v}_1\cdot \vec{v}_2=-k,\] 
and \[\vec{n}_2\cdot\vec{v}_1=\vec{n}_1\cdot\vec{v}_2=\cos(\frac{\pi}2-\alpha)=\cos(\theta-\frac\pi2)=\sin(\theta)=k^*.\]

We have
$\hat{\vec{n}}_1=p^*\vec{n}_1+p\vec{v}_1$ and $\hat{\vec{n}}_2=p^*\vec{n}_2+p\vec{v}_2$.  Thus
\begin{align*}
	\hat{\vec{n}}_1\cdot\hat{\vec{n}}_2&=(p^*)^2\vec{n}_1\cdot \vec{n}_2+2p^*p\vec{v}_1\cdot\vec{n}_2+p^2\vec{v}_1\cdot\vec{v}_2\displaybreak[0]\\
	&=-(p^*)^2k+2p^*pk^*+p^2k\displaybreak[0]\\
	&=k(p^2-(p^*)^2)+2p^*pk^*\displaybreak[0]\\
	&=k(p^2-(1-p^2))+2p^*pk^*\displaybreak[0]\\
	&=2p^2k-k+2p\sqrt{1-p^2}\sqrt{1-k^2}\displaybreak[0]\\
	&=-f(k,p).
\end{align*}
And thus $\hat{\vec{v}}_1\cdot\hat{\vec{v}}_2=f(k,p)$.
\end{proof}

\subsection{Proof of Eqn.~(\ref{eqn:qideal})}
We know that $\vec{p}_0,\ldots, \vec{p}_{d}$ are  equidistant, with $|\vec{p}_j-\vec{p}_l|=1$ for $j<l$.  Let $\tilde{\vec{p}}_j=\vec{p}_j-\vec{p}_{0}$.  Then for $j\neq l$, we have 
	$1=|\vec{p}_j-\vec{p}_l|^2=|\tilde{\vec{p}}_j-\tilde{\vec{p}}_l|^2=|\tilde{\vec{p}}_j|^2+ |\tilde{\vec{p}}_l|^2-2\tilde{\vec{p}}_j\cdot\tilde{\vec{p}}_l$ hence
	\[\tilde{\vec{p}}_j\cdot\tilde{\vec{p}}_l=\frac12.\]
	The center of mass of $\sigma_i$ is $\bar{\vec{p}}=\frac{\vec{p}_0+\cdots+\vec{p}_{i}}{i+1}$ and the distance square from $\bar p$ to $p_{n+1}$ is 
\begin{align*}
	 |\bar{\vec{p}}-\vec{p}_{0}|^2&=|\bar{\vec{p}} -(i+1)\frac{\vec{p}_{0}}{i+1}|^2\\
	&=\bigl|\frac{(\vec{p}_0-\vec{p}_{0})+\cdots+ (\vec{p}_i-\vec{p}_{0})}{i+1}\bigr|^2\\
	&=\frac{1}{(i+1)^2}\sum_{1\leq j,l\leq i} \tilde{\vec{p}}_j\cdot\tilde{\vec{p}}_l\\
	&=\frac{1}{(i+1)^2}\bigl(\frac{i(i-1)}2+i\bigr)\\
	&=\frac{i}{2(i+1)}.
\end{align*}

For $j=(i+1),\ldots,d$, the triangle $\vec{p}_0\bar{\vec{p}}\vec{p}_j$ has a right angle at $\bar{\vec{p}}$.  Hence we must have $r_i=|\vec{p}_j-\bar{\vec{p}}|=\sqrt{1-|\vec{p}_0-\bar{\vec{p}}|^2}=\sqrt{1-\frac{i}{2(i+1)}}=\sqrt{\frac{i+2}{2(i+1)}}$.

Let $\Conv{\vec{p}_{i+1}-\bar{\vec{p}},\ldots,\vec{p}_{d-1}-\bar{\vec{p}}}$
and $\Conv{\vec{p}_{i+2}-\bar{\vec{p}},\ldots,\vec{p}_{d}-\bar{\vec{p}}}$
be two $(d-i-1)$-dimensional faces of the $(d-i)$-dimensional simplex that is the convex hull of $\bar{\vec{p}},\vec{p}_{i+1},\ldots,\vec{p}_d$.  Let $\vec{v}_1$ and $\vec{v}_2$ be unit vectors perpendicular to  these faces, facing in the simplex.  

For notational simplicity, let's relabel the vectors:  let $\hat{\vec{p}}_j\equiv\vec{p}_{i+j}-\bar{\vec{p}}$ and $m\equiv d-i$.  So now our vertices are $\hat{\vec{p}}_1,\ldots,\hat{\vec{p}}_m$ and $\vec{0}$.   Let $s_i\equiv \frac{1}{2i+2}$.  We have $\|\hat{\vec{p}}_j\|=r_i$ and for $j\neq l$,
since $1=\|\hat{\vec{p}}_j-\hat{\vec{p}}_l\|^2=\|\hat{\vec{p}}_j\|^2+\|\hat{\vec{p}}_l\|^2-2\hat{\vec{p}}_j\cdot \hat{\vec{p}}_l$, \[\hat{\vec{p}}_j\cdot\hat{\vec{p}}_l=r_i^2-\frac12=s_i.\]

Using the analog of the cross-product in higher dimension, we  let 
\begin{align*}
	\vec{w}_1&\equiv*\hat{\vec{p}}_1\wedge\cdots\wedge \hat{\vec{p}}_{m-1},\\
	\vec{w}_2&\equiv(-1)^{m-1}*\hat{\vec{p}}_2\wedge\cdots\wedge \hat{\vec{p}}_{m}.
\end{align*}
Then $\vec{v}_j=\vec{w}_j/\|\vec{w}_j\|$ if the vectors $\hat{\vec{p}}_1,\ldots,\hat{\vec{p}}_m$ form a basis compatible with the orientation. We have
\begin{align*}
\|\vec{w}_1\|^2=\|\vec{w}_2\|^2&=\det(\hat{\vec{p}}_j\cdot \hat{\vec{p}}_l)_{2\leq j,l\leq m}\\
&=\det\bigl(\delta_{jl}r_i^2+(1-\delta_{jl})s_i\bigr)_{2\leq j,l\leq m}\\
&=\bigl(r_i^2+(m-2)s_i\bigr)(r_i^2-s_i)^{m-2},
\end{align*}
while
\[\vec{w}_1\cdot \vec{w}_2=(-1)^{m-1}\det(\hat{\vec{p}}_j\cdot \hat{\vec{p}}_{l+1})_{1\leq j,l\leq m-1}=-(r_i^2-s_i)^{m-2}s_i.\]
So
\begin{align*}
\vec{v}_1\cdot\vec{v}_2&=\frac{\vec{w}_1\cdot\vec{w}_2}{\|\vec{w}_1\|\|\vec{w}_2\|}\\
&=-\frac{s_i}{r_i^2+(m-2)s_i}=-\frac1{d}.	
\end{align*}

\section{Statistics for  Poisson--Voronoi tessellations in high dimension}
\label{app:Poisson}

\subsection{Dimension 5}
\begin{theorem}Let $X$ be a Poisson--Voronoi tessellation of intensity $\gamma$ in $\R^5$.  Then
\begin{align*}
\gamma^{(0)}&={\frac {1296000}{676039}}{\pi }^{4}\gamma,\quad
\gamma^{(1)}={\frac {3888000}{676039}}{\pi }^{4}\gamma,\\
\gamma^{(2)}&={\frac {3888000}{676039}}{\pi }^{4}\gamma+\gamma^{(4)}-\gamma,\\
\gamma^{(3)}&={\frac {1296000}{676039}}{\pi }^{4}\gamma+2
\gamma^{(4)}-2\gamma,
\end{align*}
\begin{align*}
\bar{q}_{3,5}^\Poisson 
&={\frac {19440000{\pi }^{4}\gamma}{3888000{\pi }^{4}\gamma+676039\gamma^{(4)}-676039\gamma}},\\
\bar{q}_{2,5}^\Poisson 
&={\frac {12960000{\pi }^{4}\gamma}{648000{\pi }^{4}\gamma+676039\gamma^{(4)}-676039\gamma}},\\
\bar{q}_{1,5}^\Poisson 
&={\frac {19440000}{676039}}{\frac {{\pi }^{4}\gamma}{\gamma^{(4)}}},\\
\bar{q}_{0,5}^\Poisson 
&={\frac {7776000}{676039}}{\pi }^{4},\\
\bar{Z}_5^\Poisson 
&={\frac {2\gamma^{(4)}}\gamma},
\end{align*}
while in all other cases where $j\leq k$, we have
\[n_{k,j}(X)=\frac{\gamma^{(j)}}{\gamma^{(k)}}\binom{6-j}{k-j}.\]
\end{theorem}

Numerical simulations provide an estimate for the value of $\bar{q}_{3,5}^\Poisson$.  Using this estimate, we find an estimate  $\gamma^{(4)}\simeq 44.20(1)\gamma$.  Substituting this value in the quantities above yield estimates for the other coordination statistics agreeing with the simulated value with a relative error of $0.5\%$.  By contrast, using the known value of $\bar{q}_{3,5}^\ideal$ to perform the same trick yields quantities agreeing with relative errors $16\%$ for $\bar{q}_{1,5}^\Poisson$ and  $13\%$ for $\bar{Z}_5^\Poisson$.

\subsection{Dimension 6}
\begin{theorem}Let $X$ be a Poisson--Voronoi tessellation of intensity $\gamma$ in $\R^6$.  Then
\begin{align*}
\gamma^{(0)}&={\frac {12964479}{10000}}\gamma,\quad \gamma^{(1)}={\frac {90751353}{20000}}\gamma,\displaybreak[0]\\
\gamma^{(3)}&=\frac{5\gamma^{(2)}}2-{\frac {90751353}{8000}}\,\gamma,\displaybreak[0]\\
\gamma^{(4)}&=2\gamma^{(2)}-{\frac {220356143}{20000}}\gamma,\displaybreak[0]\\
\gamma^{(5)}&=\frac{\gamma^{(2)}}2-{\frac {116560311}{40000}}\gamma
\end{align*}
\begin{align*}
\bar{q}_{4,6}^\Poisson&=\frac{272254059\gamma}{10000\gamma^{(2)}},\displaybreak[0]\\
\bar{q}_{3,6}^\Poisson&=\frac{363005412\gamma}{20000\gamma^{(2)}-90751353\gamma},\displaybreak[0]\\
\bar{q}_{2,6}^\Poisson&=\frac{907513530\gamma}{40000\gamma^{(2)}-220356143\gamma},\displaybreak[0]\\
\bar{q}_{1,6}^\Poisson&=\frac{1089016236\gamma}{20000\gamma^{(2)}-116560311\gamma},\displaybreak[0]\\
\bar{q}_{0,6}^\Poisson&=\frac{90751353}{10000},\displaybreak[0]\\
\bar Z_6^\Poisson&=\frac{20000\gamma^{(2)}-116560311\gamma}{20000\gamma},
\end{align*}
while in all other cases where $j\leq k$, we have
\[n_{k,j}(X)=\frac{\gamma^{(j)}}{\gamma^{(k)}}\binom{7-j}{k-j}.\]
\end{theorem}

Numerical simulations provide an estimate for the value of $\bar{q}_{4,6}^\Poisson$.  Using this estimate, we find an estimate  $\gamma^{(2)}\simeq6030(1)\gamma$.  Substituting this value in the quantities above yield estimates for the other coordination statistics agreeing with the simulated value with a relative error of $0.3\%$.  By contrast, using the known value of $\bar{q}_{4,6}^\ideal$ to perform the same trick yields quantities agreeing with relative errors $25\%$ for $\bar{q}_{1,6}^\Poisson$ and  $20\%$ for $\bar{Z}_5^\Poisson$.

\subsection{Dimension 7}
\begin{theorem}Let $X$ be a Poisson--Voronoi tessellation of intensity $\gamma$ in $\R^7$.  Then
\begin{align*}
\gamma^{(0)}&={\frac {3442073600000}{322476036831}}{\pi }^{6}\gamma,\displaybreak[0]\\
\gamma^{(1)}&={\frac {137}{322476036831}}{\pi }^{6}\gamma,\displaybreak[0]\\
\gamma^{(3)}&=-{\frac {48189030400000}{322476036831}}{\pi }^{6}\gamma+3\gamma^{(2)},\displaybreak[0]\\
\gamma^{(4)}&=-{\frac {58515251200000}{322476036831}}{\pi }^{6}\gamma+3\gamma^{(2)}+\gamma^{(6)}-\gamma,\displaybreak[0]\\
\gamma^{(5)}&=-{\frac {6884147200000}{107492012277}}{\pi }^{6}\gamma+\gamma^{(2)}+2\gamma^{(6)}-2\gamma,
\end{align*}
\begin{align*}
\bar{q}_{5,7}^\Poisson&={\frac {96378060800000}{322476036831}}\,{\frac {{\pi }^{6}\gamma}{\gamma^{(2)}}},\displaybreak[0] \\
\bar{q}_{4,7}^\Poisson&={\frac {192756121600000{\pi }^{6}\gamma}{967428110493\gamma^{(2)}-48189030400000{\pi }^{6}\gamma}},
\end{align*}
\begin{widetext}
\begin{align*}
\bar{q}_{3,7}^\Poisson&={\frac {240945152000000{\pi }^{6}\gamma}{967428110493\gamma^{(2)}+322476036831\gamma^{(6)}-(58515251200000{\pi }^{6}+322476036831)\gamma}},\displaybreak[0]\\
\bar{q}_{2,7}^\Poisson&={\frac {192756121600000{\pi }^{6}\gamma}{3\bigl(107492012277\gamma^{(2)}+214984024554\gamma^{(6)}-(6884147200000{\pi }^{6}+214984024554)\gamma\bigr)}},\displaybreak[0]\\
\end{align*}
\end{widetext}
\begin{align*}
\bar{q}_{1,7}^\Poisson&={\frac {96378060800000}{322476036831}}{\frac {{\pi }^{6}\gamma}{\gamma^{(6)}}},\displaybreak[0]\\
\bar{q}_{0,7}^\Poisson&={\frac {27536588800000}{322476036831}}{\pi }^{6},\displaybreak[0]\\
\bar{Z}_7^\Poisson&=\frac{2\gamma^{(6)}}{\gamma},
\end{align*} 
while in all other cases where $j\leq k$, we have
\[n_{k,j}(X)=\frac{\gamma^{(j)}}{\gamma^{(k)}}\binom{8-j}{k-j}.\]
\end{theorem}

Numerical simulations provide an estimate for the value of $\bar{q}_{5,7}^\Poisson$ and $\bar{q}_{3,7}^\Poisson$.  Using these estimates, we find the estimates  $\gamma^{(2)}\simeq6487(2)\times 10\gamma$ and $\gamma^{(6)}\simeq227.9(1)\gamma$.  Substituting these values in the quantities above yield estimates for the other coordination statistics agreeing with the simulated value with a relative error of $0.4\%$.

\subsection{Dimension 8}
\begin{theorem}Let $X$ be a Poisson--Voronoi tessellation of intensity $\gamma$ in $\R^8$.  Then
\begin{align*}
\gamma^{(0)}&={\frac {4155610296921974}{45956640625}}\gamma,\displaybreak[0]\\
\gamma^{(1)}&={\frac {18700246336148883}{45956640625}}\gamma,\displaybreak[0]\\
\gamma^{(3)}&=-{\frac {12466830890765922}{6565234375}}\gamma+\frac{7\gamma^{(2)}}2,\displaybreak[0]\\
\gamma^{(5)}&={\frac {37400492672297766}{6565234375}}\gamma-{\frac {35}{4}}\gamma^{(2)}+\frac{5\gamma^{(4)}}2,\displaybreak[0]\\
\gamma^{(6)}&={\frac {257647930322443638}{45956640625}}\gamma-\frac{17\gamma^{(2)}}2+2\gamma^{(4)},\displaybreak[0]\\
\gamma^{(7)}&={\frac {68567707769134446}{45956640625}}\gamma-\frac{9\gamma^{(2)}}{4}+\frac{\gamma^{(4)}}2
\end{align*}
\begin{align*}
\bar{q}_{6,8}^\Poisson&={\frac {149601970689191064}{45956640625}}{\frac \gamma{\gamma^{(2)}}},\displaybreak[0]\\
\bar{q}_{5,8}^\Poisson&=\frac {99734647126127376\gamma}{45956640625\gamma^{(2)}-24933661781531844\gamma},\displaybreak[0]\\
\bar{q}_{4,8}^\Poisson&=\frac {74800985344595532\gamma}{6565234375\gamma^{(4)}},
\end{align*}
\begin{widetext}
\begin{align*}
\bar{q}_{3,8}^\Poisson&=\frac {299203941378382128\gamma}{149601970689191064\gamma-229783203125\,\gamma^{(2)}+65652343750\gamma^{(4)}},\displaybreak[0]\\
\bar{q}_{2,8}^\Poisson&=\frac {698\gamma}{515295860644887276\gamma-781262890625\,\gamma^{(2)}+183826562500\gamma^{(4)}},\displaybreak[0]\\
\bar{q}_{1,8}^\Poisson&=\frac {598407882756764256\gamma}{274270831076537784\gamma-413609765625\,\gamma^{(2)}+91913281250},\displaybreak[0]\\
\bar{q}_{0,8}^\Poisson&=\frac {37400492672297766}{45956640625},\displaybreak[0]\\
\bar{Z}_8^\Poisson&=\frac{274270831076537784\gamma-413609765625\gamma^{(2)}+91913281250\gamma^{(4)}}{91913281250\gamma},
\end{align*}
\end{widetext}
while in all other cases where $j\leq k$, we have
\[n_{k,j}(X)=\frac{\gamma^{(j)}}{\gamma^{(k)}}\binom{9-j}{k-j}.\]
\end{theorem}

Numerical simulations provide an estimate for the value of $\bar{q}_{2,8}^\Poisson$ and $\bar{Z}_{8}^\Poisson$.  Using these estimates, we find the estimates  $\gamma^{(2)}\simeq745(1)\times 1000\gamma$ and $\gamma^{(4)}\simeq370(1)\times 1000\gamma$.  Substituting these values in the quantities above yield estimates for the other coordination statistics agreeing with the simulated value with a relative error of $0.6\%$.

\section{Relation to the blocking (disconnected) and connected correlation functions of fluids in disordered porous media}
\label{sec:blocking}
In the statistical mechanics of fluids in a disordered porous medium, and more generally in the theory of disordered systems, the existence of two different types of averages, a ``quenched'' average over disorder and an ``annealed'' average over the equilibrated fluid, requires the introduction of several  distinct correlation functions. At the pair level, these correlation functions are known as ``connected'' and ``blocking''\cite{given:1992,lomba:1993} (or ``disconnected''). For a fluid they are defined according to
\begin{equation}
\begin{aligned}
\label{eq_connected}
\rho^{(2)}_{c}(\vert \mathbf r-\mathbf{r'} \vert)=\overline{\left\langle \rho(\mathbf r) \rho(\mathbf{r'})\right\rangle - \left\langle \rho(\mathbf r)\right\rangle \left\langle \rho(\mathbf{r'})\right\rangle}
\end{aligned}
\end{equation}
\begin{equation}
\begin{aligned}
\label{eq_blocking}
\rho^{(2)}_{b}(\vert \mathbf r-\mathbf{r'} \vert )= \overline{\left\langle \rho(\mathbf r)\right\rangle \left\langle \rho(\mathbf{r'})\right\rangle}
\end{aligned}
\end{equation}
where $\rho(\mathbf r)$ is the microscopic fluid density within the open volume left by the disordered environment and the brackets denote an equilibrium in the presence of the latter.

In the context of partially pinned fluid configurations, the same definitions apply provided that one associates the average in the presence of disorder with the conditional average $\left\langle \; \right\rangle_{\mathcal B}$, the microscopic fluid density with $\hat\rho(\mathbf r)$ (and thus the mean fluid density with $(1-c)\rho$), and the average over the disordered environment  with the double average over the pinned set and the equilibrium reference configuration. For instance, the blocking pair density function then reads
\begin{equation}
\begin{aligned}
\label{eq_blocking_pinned}
\rho^{(2)}_{b}(\vert \mathbf r-\mathbf{r'} \vert )= \left\langle \overline{\left\langle \hat\rho(\mathbf r)\right\rangle_{\mathcal B} \left\langle \hat\rho(\mathbf{r'})\right\rangle_{\mathcal B}} \right\rangle \, .
\end{aligned}
\end{equation}

It is now easy to prove by following the line of reasoning as for deriving Eq. (\ref{eq_average_condit_rho}) (see also Ref.~\onlinecite{krakoviack:2010}) that
\begin{equation}
\begin{aligned}
\label{eq_blocking_pinned_equiv}
\left\langle \overline{\left\langle \hat\rho(\mathbf r)\right\rangle_{\mathcal B} \left\langle \hat\rho(\mathbf{r'})\right\rangle_{\mathcal B}} \right\rangle 
=\left\langle \overline{\hat\rho(\mathbf r) \left\langle \hat\rho(\mathbf{r'})\right\rangle_{\mathcal B}} \right\rangle \, ,
\end{aligned}
\end{equation}
so that the overlap $Q_c(\infty)$ given in Eq. (\ref{eq_overlap_infty}) can also be written as
\begin{equation}
\label{eq_overlap_infty_block}
\begin{aligned}
&Q_c(\infty)=
\\&\frac{1}{(1-c)^2N} \int d\mathbf{ r}\int d\mathbf{ r'}\, \Theta(a-|\mathbf{r}-\mathbf{r'}|)  \rho^{(2)}_{b}(\vert \mathbf r-\mathbf{r'} \vert ).
\end{aligned}
\end{equation}

After introducing the blocking total correlation function through $\rho^{(2)}_{b}( r)=(1-c)^2\rho^2[h_b( r)+1]$, one finally obtains that
\begin{equation}
\label{eq_overlap_infty_block_final}
\begin{aligned}
&Q_c(\infty)-Q_0(\infty)=\rho \int d\mathbf{ r}\, \Theta(a-r) h_b( r).
\end{aligned}
\end{equation}
The quantity that is relevant for extracting a point-to-set correlation length associated with the slowdown of dynamics is therefore simply related to the integral over a sphere of radius $a$ of the blocking total correlation function. (Note that even for hard spheres the latter does not equal $-1$ inside the core.) One could then envisage to use the approximations from liquid-state theory formulated in the context of fluids in disordered porous media~\cite{given:1992,lomba:1993,rosinberg:1994,krakoviack:2001} to compute the blocking function, although one may wonder if simple approximations, e.g., HNC or Percus-Yevick~\cite{hansen:1986}, would be able to capture the behavior of the overlap. In addition, it should be recalled that whereas the solid porous matrix that leads to the average over the disorder does not change with the thermodynamic point, the disorder associated with a pinned sets of particles in reference equilibrium fluid configurations does vary with packing fraction (and temperature in general).


\begin{thebibliography}{75}
\bibitem [{foo({\natexlab{a}})}]{footnote:1}%
  \BibitemOpen
  \href@noop {} {} \bibinfo {note} {In order to keep the
  nomenclature consistent, we should thus refer to the ``hard spheres'' as
  ``hard balls'', but we instead follow the standard physical
  convention.}\BibitemShut {Stop}%
\bibitem [{foo({\natexlab{b}})}]{footnote:2}%
  \BibitemOpen
  \href@noop {} {} \bibinfo {note} {The name comes from
  $d$=3, where the quantity measures the number of tetrahedra sharing a common
  bond, and thus forming a structure akin to the shaft traditionally used for
  spinning fibers.}\BibitemShut {Stop}%
\end{thebibliography}
%

\end{document}